%% file: randomized.tex
\documentclass[11pt, onecolumn]{article}
\usepackage[total={6.5in,9.0in},centering]{geometry}
\usepackage{setspace}
\usepackage{authblk}
\usepackage{verbatim}

\usepackage{ifpdf}
\usepackage{cite}

\usepackage[dvips]{graphicx}
\DeclareGraphicsExtensions{.eps}

\usepackage[cmex10]{amsmath}
\usepackage{latexsym}
\usepackage{amssymb}
\usepackage{array}

\usepackage[labelfont=bf]{caption}
\usepackage[font=small]{subfig}

\usepackage{mdwmath}
\usepackage{mdwtab}
\usepackage{fixltx2e}
\usepackage{stfloats}
\usepackage[amsmath,amsthm,thmmarks]{ntheorem}
\usepackage{paralist}
\usepackage{pstricks}
\usepackage{epsfig}
\usepackage[inline]{enumitem}
\usepackage{url}
\usepackage{multirow}
\usepackage{multicol}
\usepackage{array}
\usepackage{xparse}
\usepackage{nicefrac}
\usepackage{array}

\usepackage[boxed, linesnumbered, longend]{algorithm2e}
\SetAlFnt{\fontsize{9}{10}\selectfont}
\SetAlCapFnt{\small}
\SetAlCapNameFnt{\small}
\usepackage{etoolbox}
\AtBeginEnvironment{algorithm}{\setstretch{1.0}}
\SetAlCapSkip{1em}

\input{useful}
\input{custom}

\usepackage{tikz}
\usepackage{pgfplots}
\pgfplotsset{width=7cm}
\usetikzlibrary{calc}
\usepgfplotslibrary{groupplots}

\newtheorem{theorem}{Theorem}
\newtheorem{definition}{Definition}
\newtheorem{lemma}[theorem]{Lemma}
\newtheorem{assumption}{Assumption}
\newtheorem*{remark}{Remark}

\begin{document}

\title{\textbf{Robust Neighbor Discovery in Multi-Hop Multi-Channel Heterogeneous Wireless Networks}\footnote{A preliminary version of the work has appeared earlier in the Proceedings of the 31st International Conference on Distributed Computing Systems (ICDCS), 2011~\cite{MitZen+:2011:ICDCS}.}}

\author{Yanyan Zeng\footnote{ Yanyan Zeng is currently working at JPMorgan Chase.}}
\author{K. Alex Mills}
\author{Shreyas Gokhale}
\author{Neeraj Mittal\thanks{This work was supported, in part, by the National Science Foundation (NSF) under grant number CNS-1115733.}}
\author{S. Venkatesan}
\author{\\R. Chandrasekaran}
\affil{
Department of Computer Science \\
The University of Texas at Dallas \\
Richardson, TX 75080, USA \\
E-mails: yanyan.bansal@gmail.com \quad k.alex.mills@utdallas.edu \quad shreyas.gokhale@utdallas.edu \quad neerajm@utdallas.edu \quad  venky@utdallas.edu \quad chandra@utdallas.edu
}

\date{}

\maketitle


%

\input{abstract}

\input{introduction}

\input{model}

\input{synchronous}

\input{asynchronous}

\input{summary}

\input{extensions}

\input{relatedWork}

\input{conclusion}

\bibliographystyle{plain}
\bibliography{citations}

\begin{appendix}
\input{appendix}
\end{appendix}
\end{document}

%% file: useful.tex
\newcommand{\ang}[2][{}]{{#1\langle} #2 {#1\rangle}}
\newcommand{\ceil}[2][{}]{{#1\lceil} #2 {#1\rceil}}

\newcommand{\card}[2][{}]{{#1\lvert} #2 {#1\rvert}}

\newcommand{\true}{\textsf{true}}

\newcommand{\secref}[1]{Section~\ref{sec:#1}}
\newcommand{\figref}[1]{Fig.~\ref{fig:#1}}

\newcommand{\thmref}[1]{Theorem~\ref{thm:#1}}
\newcommand{\lemref}[1]{Lemma~\ref{lem:#1}}

\newcommand{\algoref}[1]{Algorithm~\ref{alg:#1}}
\newcommand{\enumref}[1]{(\ref{enum:#1})}
\newcommand{\asmref}[1]{Assumption~\ref{asm:#1}}

\newcommand{\Tabref}[1]{Table~\ref{tab:#1}}

\newcommand{\Algoref}[1]{Algorithm~\ref{alg:#1}}

%% file: custom.tex
\newcommand{\model}{M$^2$HeW}

\NewDocumentCommand\acs{ g }{ \IfNoValueTF{#1} {\mathcal{A}} {\mathcal{A}(#1)} }

\newcommand{\clockof}[1]{C_{#1}}
\newcommand{\clockat}[2][]{C_{#1}(#2)}

\newcommand{\drmaximum}{\delta}

\newcommand{\destimate}{\Delta_{est}}

\newcommand{\dmaximum}{\Delta}

\newcommand{\acsmaximum}{S}

\newcommand{\epmaximum}{\epsilon}
\newcommand{\fpmaximum}{\phi}

\newcommand{\nsize}{N}

\newcommand{\degreeon}[2]{\Delta(#1,#2)}

\newcommand{\link}[2]{(#1, #2)}

\newcommand{\incoming}[2]{\mathcal{N}(#1,#2)}

\newcommand{\lspan}[2]{span(#1,#2)}

\newcommand{\srminimum}{\rho}

\newcommand{\A}[2]{A(#1,#2)}

\newcommand{\B}[2]{B(#1,#2)}

\newcommand{\C}[2]{C(#1,#2)}

\newcommand{\D}[2]{D(#1,#2)}
\newcommand{\E}[2]{E(#1,#2)}

\newcommand{\F}[2]{F(#1,#2)}

\newcommand{\ABC}[2]{ABC(#1,#2)}

\newcommand{\EBAR}[3][]{F(#2,#3,#1)}

\newcommand{\AF}[3]{\hat{A}(#1,#2,#3)}

\newcommand{\BF}[3]{\hat{B}(#1,#2,#3)}

\newcommand{\CF}[3]{\hat{C}(#1,#2,#3)}

\newcommand{\AFS}[3]{\hat{A^\prime}(#1,#2,#3)}
\newcommand{\CFS}[3]{\hat{C^\prime}(#1,#2,#3)}

\newcommand{\DF}[3]{\hat{D}(#1,#2,#3)}

\newcommand{\EF}[3]{\hat{E}(#1,#2,#3)}

\newcommand{\ABCF}[3]{\hat{ABC}(#1,#2,#3)}

\newcommand{\EBARF}[4][]{\hat{F}(#2,#3,#4,#1)}

\newcommand{\flength}{L}

\newcommand{\slotu}[1]{b_{#1}}
\newcommand{\slotv}[1]{a_{#1}}

\newcommand{\admissible}{admissible}

\newcommand{\under}{\sqsubset}

\newcommand{\precedence}{precedence}

\newcommand{\overlapAll}[1]{overlap\text{-}all(#1)}
\newcommand{\overlap}[2]{overlap(#1,#2)}

\newcommand{\node}[1]{node(#1)}

\newcommand{\legal}{aligned}

\newcommand{\lpair}[2]{\ang{#1,#2}}

\newcommand{\bsize}{B}
\newcommand{\bminimum}{W}

\newcommand{\phase}[2]{\Lambda(#1,#2)}
\newcommand{\critical}{potent}

\newcommand{\starttime}{T_s}
\newcommand{\finishtime}{T_{\!f}}

\newcommand{\vmaximum}{\Theta}

\newcommand{\framesuntil}[2]{F_{#1,#2}}

\NewDocumentCommand\starttimeOf{ g g g }{ \IfNoValueTF{#1} {T_b\xspace} { \IfNoValueTF{#3} {T_b^{\max}(#1,#2)\xspace} {T_b^{#3}(#1,#2)\xspace}} }
\NewDocumentCommand\finishtimeOf{ g g g }{ \IfNoValueTF{#1} {T_e\xspace} { \IfNoValueTF{#3} {T_e^{\min}(#1,#2)\xspace} {T_e^{#3}(#1,#2)\xspace}} }

\newcommand{\oldchannel}{j}
\newcommand{\busychannels}{C}

\newcommand{\given}{\mid}

\newcommand{\eulernum}{\mathrm{e}}

\newcommand{\slot}[2]{slot(#1,#2)}

\newcommand{\sdf}[2]{\sigma(#1,#2)}

%% file: abstract.tex
\begin{abstract}
An important first step when deploying a wireless ad hoc network is \emph{neighbor discovery} in which every node attempts to determine the set of nodes it can communicate within one wireless hop. 
In the recent years, \emph{cognitive radio (CR)} technology has gained attention as an attractive approach to alleviate  spectrum congestion. A cognitive radio transceiver
can operate over a wide range of frequencies possibly spanning multiple frequency bands. A cognitive radio node can opportunistically utilize unused wireless spectrum without interference from other wireless devices in its vicinity. Due to spatial variations in frequency usage and hardware variations in radio transceivers, different nodes in the network may perceive different subsets of frequencies available to them for communication. This \emph{heterogeneity} in the available channel sets across the network increases the complexity of solving the neighbor discovery problem in a cognitive radio network. In this work, we design and analyze several randomized algorithms for neighbor discovery in such a (heterogeneous) network under a variety of  assumptions (\emph{e.g.}, maximum node degree known or unknown) 
for both synchronous and asynchronous systems under minimal knowledge. 
We also show that our randomized algorithms are naturally suited to tolerate unreliable channels and adversarial attacks. 
\\

\noindent
\textbf{Key words:}
multi-hop multi-channel wireless network, cognitive radio technology, heterogeneous channel availability, neighbor discovery,  randomized algorithm, asynchronous system, clock drift, lossy channel, jamming attack
\end{abstract}


%% file: introduction.tex

\section{Introduction}\label{sec:introduction}

Neighbor discovery is an important step in forming a \emph{self-organizing} wireless ad hoc network without any support from an existing communication 
infrastructure~\cite{McGBor:2001:MOBIHOC,VasKur+:2005:INFOCOM,BorEph+:2007:AN,VasTow+:2009:MOBICOM,WilKar+:2010:ADHOCNOW,SonXie:2012:INFOCOM,
GonCha+:2013:MASS,BiaPar:2013:TMC,CheBia+:2014:MOBIHOC:a,CheBia+:2014:MOBIHOC:b,RusVas+:2014:TPDS}.
When deployed, nodes initially have no prior knowledge of other nodes that they can communicate with directly.


The results of neighbor discovery can  be used to solve other important communication problems such as
medium access control~\cite{BaoGar:2002:MOBICOM,ChoYan+:2002:MOBICOM}, routing~\cite{CesCuo+:2010:AN}, broadcasting~\cite{SonXie:2012:INFOCOM,RehVia+:2013:ComCom},
clustering~\cite{LinGer:1997:JSAC,HeiCha+:2000:ICSS,ChaHsu:2000:MNA},
collision-free scheduling~\cite{GanDaw+:2008:JPDC,GanZha+:2008:CN}, spanning tree construction~\cite{ChoKha+:2009:JSAC}, and topology control~\cite{WatLi+:2001:INFOCOM,LiHal+:2005:TN}.
Many algorithms for solving these problems implicitly assume that all nodes know their one-hop and sometimes even two-hop neighbors.
Many location-based routing protocols (\emph{e.g.} localized routing in Vehicular Ad hoc Networks (VANETs)) use the position of neighboring nodes to make 
routing/forwarding decisions. The neighborhood information is also used to update the reachability status of nodes. 
A better neighbor discovery algorithm, which uses fewer messages and has higher accuracy, can be used to improve the performance of location-based routing  algorithms~\cite{BouRez+:2009:GLOBECOM,Bou:2005:Book,Bou:2008:Book}.
Neighborhood information helps reduce the cost of flooding in multicast tree construction using flooding algorithms~\cite{LimKim:2000:MSWiM}.
Mobile sensing applications, ranging from mobile social networking to proximity-based gaming, involve collection and sharing of data among nearby users. The success of these applications depends on neighborhood information~\cite{CheBia+:2014:MOBIHOC:b}.
Neighbor discovery is extremely important in Underwater Acoustic (UWA) Networks and needs to be done frequently because nodes may move proactively due to the unpredictable underwater currents~\cite{ZhaZha:2010:INFOCOM}.
More details of how the results of neighbor discovery can be used to solve other communication problems  can be found in~\cite{Bou:2005:Book,Bou:2008:Book}.

Cognitive radio (CR) technology has emerged as a promising approach for improving spectrum utilization efficiency and meeting the increased demand for wireless communications \cite{BroWol+:Corvus}. 
A CR node can scan a part of the wireless spectrum, and identify unused or underutilized channels in the spectrum \cite{YucArs:2009:CST,AkyLo+:2011:PC}. CR nodes in a network can then use these channels
opportunistically for communication among themselves even if the channels belong to licensed users. The licensed users are referred to as the \emph{primary users}, and CR nodes are referred to as the \emph{secondary users}. (Of course, when a primary user arrives and starts using its channel, the secondary users have to vacate the channel.)
Due to spatial variations in frequency usage, hardware variations in radio transceivers and uneven propagation of wireless signals, 
different nodes in the network may perceive different subsets of frequencies available to them for communication. 
This gives rise to a \emph{multi-hop, multi-channel, heterogeneous wireless network}, abbreviated as \emph{\model{} network}. 
The focus of this work is on solving the neighbor discovery problem in an \model{} network.

A large number of neighbor discovery algorithms have been proposed in the literature.
Most of the algorithms suffer from one or more of the following limitations: 
\begin{inparaenum}[(i)]
\item all nodes are assumed to be synchronized (synchronous system), 
\item the entire network is assumed to operate on a single channel (single-channel network),
\item all nodes are assumed to be able to communicate with each other (single-hop network),
\item all channels are assumed to be available to all nodes (homogeneous network),  or
\item the algorithm is only evaluated experimentally (no theoretical guarantees provided).
\end{inparaenum}
A more detailed discussion of the related work can be found in \secref{related}.



\paragraph{Our Contributions:} 
Our main contribution in this work is \emph{two} randomized neighbor discovery algorithms for 
an \emph{\model{} network} when the system is  \emph{asynchronous} that guarantee success with arbitrarily high probability. The first 
algorithm assumes that nodes know a good upper bound on the maximum degree of 
any node in the network. The second algorithm does not make any such 
assumption. Both algorithms only assume that the maximum drift rate of the  clock of any node is bounded, with the second algorithm assuming a tighter but unknown bound whose value depends on various system parameters. 
None of the algorithms require clocks of different nodes to be synchronized. In fact, clocks of any two nodes may have arbitrary  skew with respect to each  other. Other advantages of our algorithms are as follows:
\begin{inparaenum}[(i)]
\item nodes do not need to agree on a universal channel set, and
\item the running time of an algorithm depends on the ``degree of heterogeneity'' in the network; the running time decreases as the available channel sets become more homogeneous.
\end{inparaenum}

Our algorithms for an asynchronous system are based on those for a synchronous system.
Therefore, as additional contributions, we also present a suite of randomized neighbor discovery algorithms for an \emph{\model{} network} when the system is \emph{synchronous} under a variety of assumptions such as: 
\begin{inparaenum}[(i)]
\item whether all the nodes start executing the neighbor discovery algorithm at the same time or not, and 
\item whether nodes are aware of an estimate on the upper bound on the maximum degree of any node in the network or not.
\end{inparaenum}

We believe that our approach for transforming a \emph{state-less} algorithm developed for a synchronous \model{} network to work for an asynchronous \model{} network can also be applied to other important communication problems in an \model{} network with running time increasing by only a constant factor.

We show that our randomized algorithms can easily tolerate unreliable channels.
We also prove that, our algorithms, with minor modification in the asynchronous case, 
are tolerant to jamming attacks by a reactive but ``memory-less'' jammer under certain assumption. The running time of our algorithms, when subject to a jamming attack, increases by at most a \emph{constant} factor. In fact, for sufficiently large values of system parameters 
(namely, number of nodes and number of channels), the running time  increases only by a factor of at most two in the worst-case.

\paragraph{Organization:} 
The rest of the manuscript is organized as follows. We describe our system
model for a multi-hop multi-channel heterogeneous wireless network in 
\secref{model}. For ease of exposition, we first present a suite of randomized 
neighbor discovery algorithms for a \emph{synchronous} system under a variety 
of assumptions and analyze their time complexity in \secref{synchronous}. We 
then present two randomized neighbor discovery algorithms for an 
\emph{asynchronous} 
system, which are derived from their synchronous counterparts, and analyze 
their complexity in \secref{asynchronous}.
We discuss several extensions to our algorithms to enhance their applicability and improve their robustness in 
\secref{extensions}. 
Finally, in \secref{related}, we give a comparison of our contributions with existing research and also examine other related work done on neighbor discovery.

%% file: model.tex
\section{System Model}
\label{sec:model}

We assume a multi-hop  multi-channel heterogeneous wireless (\model{}) network consisting of one or more radio nodes.
Let $\nsize$ denote the total number of radio nodes. Nodes do not know $\nsize$.  Each
radio node  is equipped with a transceiver (transmitter-receiver
pair), which is capable of operating over multiple frequencies or channels. However,
at any given time, a transceiver can operate (either transmit or
receive) over a single channel only. Further, a transceiver cannot
transmit and receive at the same time. Transceivers of different nodes need not be identical; the set of channels over which a transceiver can operate may be different for different transceivers. 

Different nodes in a network may have different sets of channels available for communication. 
For example, in a cognitive radio network (a type of \model{} network), each node can scan the frequency spectrum and identify the subset of unused or under-used portions of the spectrum even those that have been licensed to other users or organizations \cite{BroWol+:Corvus}. 
A node can potentially use such frequencies  to communicate with its neighbors until they
are reclaimed by their licensed (primary) users \cite{BroWol+:Corvus}.
Due to spatial variations in frequency usage/interference and hardware variations in radio transceivers, different nodes in the network may perceive different subsets of frequencies available to them for communication.
We refer to the subset of frequencies or channels that a node can use to communicate with its neighbors as the \emph{available channel set} of the node. 
For a node $u$, we use $\acs{u}$ to denote its available channel  set. We use $\acsmaximum$ to denote the size of the largest available channel set, that is, $\acsmaximum = \max\limits_u \card{\acs{u}}$. Note that nodes do not know $\acsmaximum$.

We say that a node $v$ is a \emph{neighbor} of node $u$ on a channel $c$ if $u$ can reliably receive any message transmitted by $v$ on $c$ provided no other node in the network is transmitting on $c$ at the same time, and vice versa. We assume that the communication graph is \emph{symmetric}  because unidirectional neighborhood relationships  are expensive and impractical 
to use in wireless networks \cite{Pra:1999:DIALM} (although our algorithms work for asymmetric communication graphs as well).
For a node $u$ and a channel $c \in \acs{u}$, we use $\degreeon{u}{c}$ to denote the number of neighbors, also known as \emph{degree}, of $u$ on $c$. We use $\dmaximum$ denote the maximum degree of any node on any channel, that is, $\dmaximum = \max\limits_u \max\limits_{c \in \acs{u}} \degreeon{u}{c}$.

Note that, if nodes $u$ and $v$ are neighbors of each other on some channel, then $u$ has to discover $v$ and $v$ has to discover $u$ separately. It is convenient to assume two separate links---one from $u$ to $v$ and another from $v$ to $u$. We use $\link{u}{v}$ to denote the link from $u$ to $v$. We refer to the set of channels on which the link $\link{u}{v}$ can operate as the \emph{span} of $\link{u}{v}$ and denote it by $\lspan{u}{v}$. Note that $\lspan{u}{v} \subseteq \acs{u} \cap \acs{v}$. We refer to the ratio of $\card{\lspan{u}{v}}$ to $\card{\acs{v}}$ as the \emph{span-ratio} of the link $\link{u}{v}$.  Note that the span-ratio of any link lies between $\frac{1}{\acsmaximum}$ and $1$. Further, the span-ratio of $\link{u}{v}$ may be different from the span-ratio of $\link{v}{u}$ (because $\card{\acs{u}}$ may be different from $\card{\acs{v}}$). We use $\srminimum$ to denote the minimum span-ratio among all links. 
Note that nodes do not know $\srminimum$. 
Intuitively, $\srminimum$ can be viewed as a measure of the \emph{degree of homogeneity} in the network---the larger the value, the more homogeneous the network is in terms of channel availability for nodes and links.
The running time of our algorithms is \emph{inversely} proportional to $\srminimum$. When all links can operate on all available channels (an assumption made frequently in the literature), $\srminimum = 1$, which minimizes the running time of our algorithms. In general, the more heterogeneous the network is, the larger is the running time of our algorithms.

If nodes $v$ and $w$ are neighbors of node $u$ on channel $c$ and both $v$ and $w$ transmit on $c$ at the same time, then their transmissions  \emph{collide} at $u$. If $u$ is listening on $c$ when $v$ and $w$ transmitting on $c$ at the same time, then $u$ hears only noise.  We do not assume that nodes can detect collisions, that is, distinguish between background noise and collision noise.

The set of channels over which an \model{} network can operate may have very different propagation characteristics. For example,  a cognitive radio transceiver may be capable of operating over a wide range of frequencies  (\emph{e.g.}, as low as 100 MHz to as high as 2.5 GHz \cite{CafCor+:2007:SDR}). Radio
waves at different frequencies undergo different angles of refraction
and therefore propagate in different ways.
As a result, two 
radio nodes $u$ and $v$ that are neighbors on some channel $c$ may not
be neighbors on another channel $d$ even if both $u$ and $v$ have
 $d$ in their available channel  sets. This may happen, for
example, if $c$ is a channel in the 900 MHz band and $d$ is a channel
in the 2.4 GHz band; a lower frequency has stronger capability to penetrate
through different materials and therefore has a longer transmission range
than a higher frequency. For ease of exposition, we first describe our algorithms assuming 
that all frequencies or channels have identical propagation characteristics. This implies that if a communication link from node $u$ to node $v$ can operate over some channel $c \in \acs{u} \cap \acs{v}$ then it can also operate over all channels in $\acs{u} \cap \acs{v}$. We later discuss how
to extend our algorithms to handle diverse propagation characteristics of different frequencies.

We use ``$\log$'' to refer to the logarithm to the base $2$ and ``$\ln$'' to refer to the natural logarithm.
In this paper, we investigate the neighbor discovery problem both when the system is synchronous and when the system is asynchronous. We now describe how we model each type of system.

\paragraph{Synchronous System:} We assume that the execution of the system is divided into
\emph{synchronized time-slots}. In each time slot, each node can be in one of the
following three modes: \begin{inparaenum}[(i)] \item \emph{transmit mode} on some channel in its
available channel set, \item \emph{receive mode} on some channel in its
available channel set, or \item \emph{quiet mode} when the transceiver is
shut-off. 
\end{inparaenum}
Nodes may or may not start at the same time.

\paragraph{Asynchronous System:} We assume that every node is equipped with  a clock.  Clocks  of different nodes are not required to be synchronized.    For a clock $C$ and time $t$, we use $\clockat{t}$ to denote the value of $C$ at $t$. For a node $u$, we use  $\clockof{u}$ to denote the clock of $u$. 
The clock of a node may have non-zero drift and  the drift rate may change over time. The drift rate of a clock $C$ at time $t$ is given by $\frac{dC}{dt} -1$. If $C$ is an ideal clock, then, $\forall t, \Delta t \geq 0$, $\clockat{t + \Delta t} - \clockat{t} = \Delta t$; thus $\frac{dC}{dt} = 1$. At any given time, different clocks may have different drift rates. Further, the drift rate of the same clock may change over time both in magnitude and sign.
We do, however, assume that the magnitude of the maximum drift rate is \emph{bounded} and let the bound be denoted by $\drmaximum$. This implies that $\forall t, \Delta t \geq 0$, 
\begin{equation}
\label{eq:drift}
(1 - \drmaximum)\Delta t \; \leq \; \clockat{t + \Delta t} - \clockat{t} \; \leq \; (1 + \drmaximum)\Delta t
\end{equation}
For an ideal clock, $\drmaximum = 0$. In practice, $\drmaximum$ is quite small (\emph{e.g.}, $10^{-6}$ seconds/second for an ordinary quartz clock). 

%% file: synchronous.tex
\section{Neighbor Discovery in Synchronous Systems}
\label{sec:synchronous}

We first describe our algorithms for  neighbor discovery  assuming that all nodes start executing the algorithm at the same time. We then relax this assumption and describe our  algorithms for neighbor discovery when nodes may start executing the algorithm at different times. Later, in \secref{asynchronous}, we use the algorithms for variable start times to derive randomized neighbor discovery algorithms for an asynchronous system.

\subsection{Identical Start Times}

We first assume that the nodes know some upper bound on maximum node degree and present a neighbor discovery algorithm. The bound need not be tight and, in fact, may be quite loose (since the dependence is logarithmic on the value of the upper bound.) But all nodes should agree on a \emph{common} upper bound. We then relax this restriction and describe an algorithm for the case when such knowledge is not available.

\subsubsection{Knowledge of Loose Upper Bound on Maximum Node Degree}
\label{sec:identical|degree}

Let $\destimate$ denote an upper bound on the maximum node degree as known to all nodes.
The execution of the algorithm is divided into \emph{stages}. Each stage consists of $\ceil{\log(\destimate)}$ time-slots. In each time-slot of a stage, a node randomly chooses a channel from its available channel set and transmits on that channel with a certain probability (and listens on that channel with the remaining probability). Specifically, in time-slot $i$ of a stage, where $1 \leq i
\leq \ceil{\log(\destimate)}$, a node $u$ transmits on the selected channel, say $c$, with probability $\min\left(\frac{1}{2}, \frac{\card{\acs{u}}}{2^i}\right)$ and listens on $c$ with the remaining probability.

\begin{algorithm}[t]
\SetKwInput{Input}{Input}
\SetKwInput{Output}{Output}
\tcp{Algorithm for node $u$} 
\BlankLine
\Input{An upper bound on maximum node degree, say $\destimate$.}
\Output{The set of neighbors along with the subset of channels that are common with the neighbor. }
\BlankLine 
\While{\true{}}{
\tcp{execute a stage}
\For{$i \leftarrow 1$ \KwTo $\ceil{\log(\destimate)}$}{
\tcp{time-slot $i$ of the current stage}
\tcp{select a channel}
$c \leftarrow$ channel selected uniformly at random from $\acs{u}$\;
\tcp{compute the transmission probability}
$p \leftarrow \min\left( \frac{1}{2}, \frac{\card{\acs{u}}}{2^i} \right)$\;
tune the transceiver to $c$\;
switch to transmit mode with probability $p$ and  receive mode with probability $1-p$\;
\uIf{(in transmit mode)}{transmit a message containing $\acs{u}$\;}
\ElseIf{(heard a clear transmission)}{
let the received message be sent by node $v$ containing set $A$\;
add $\ang{v,A \cap \acs{u}}$ to the set of neighbors\;
}
} 
}
\caption{Neighbor discovery algorithm for a synchronous system with identical start times  and knowledge of upper bound on maximum node degree.}
\label{alg:identical|degree}
\end{algorithm}

A formal description of the algorithm is shown in \algoref{identical|degree}. 
We next analyze its running time. Consider a node $u$ and let node $v$ be its neighbor on a channel $c$.   Observe that $1 \leq \degreeon{u}{c} \leq \dmaximum \leq \destimate$.  Let $k = \max( 1, \ceil{\log \degreeon{u}{c}})$. Clearly, we have:
\begin{align}
\label{eq:slot}
2^{k-1} \; \leq \; \degreeon{u}{c} \; \leq \;  2^k
\end{align}

We say that a time-slot $t$ \emph{covers} the link $\link{v}{u}$ on channel $c$ if during $t$: \begin{inparaenum}[(i)]
\item $v$ transmits on $c$, 
\label{enum:slot|1} \item $u$ listens on $c$, 
and 
\label{enum:slot|2} \item no other neighbor of $u$ transmits on $c$. 
\label{enum:slot|3}
\end{inparaenum}
The three conditions collectively ensure that $u$ receives a clear message from $v$ (on $c$ during $t$). 
We say that a stage $s$ covers the link $\link{v}{u}$ on channel 
$c$ if some time-slot of $s$ covers $\link{v}{u}$ on $c$; also,  $s$ covers $\link{v}{u}$ if
 $s$ covers $\link{v}{u}$ on some $c \in \lspan{v}{u}$. Finally, we say that a sequence of stages covers the link $\link{v}{u}$ if some stage of the sequence covers $\link{v}{u}$.

Consider a stage $s$ and let $\tau$ denote the time-slot of $s$ that satisfies \eqref{eq:slot}.
We first compute the probability that $\link{v}{u}$ is covered on $c$ during $\tau$.
Let $\A{\tau}{c}$, $\B{\tau}{c}$, $\C{\tau}{c}$ denote the three events corresponding to the three conditions \enumref{slot|1}, \enumref{slot|2} and \enumref{slot|3}, respectively, required for coverage.
Note that the three events are \emph{mutually independent}. Therefore we first compute the  probability of occurrence of events $\A{\tau}{c}$, $\B{\tau}{c}$ and $\C{\tau}{c}$ separately. We then compute the probability that $s$ covers $\link{v}{u}$. 
The steps involved in computing the probabilities are fairly standard and  have been moved to the appendix. We only state the results here. We have:
\begin{align}
\Pr\{\A{\tau}{c}\} & \geq \frac{1}{2 \max(S,\dmaximum)} \\ 
\Pr\{\B{\tau}{c}\} & \geq \frac{1}{2\card{\acs{u}}}
\label{eq:B|1} \\
\Pr\{\C{\tau}{c}\} & \geq \frac{1}{4}
\label{eq:C|1} \\
\Pr(s \text{ covers } \link{v}{u}) & \geq  \frac{\srminimum}{16 \max( \acsmaximum, \dmaximum)}
\label{eq:stage|cover}
\end{align}

Now, consider a sequence of $M = \frac{16 \max( \acsmaximum, \dmaximum )}{\srminimum} \ln \left( \frac{N^2}{\epmaximum} \right)$ stages, where $\epmaximum$ is the probability that  neighbor discovery does not complete successfully. We show that the probability that the link $\link{v}{u}$ is not covered within $M$ stages is at most $\frac{\epmaximum}{N^2}$. The steps for computing the probability are fairly standard and have been moved to the appendix. Formally,
\begin{align}
\Pr(\link{v}{u} \text{ is not covered within } M \text{ stages}) &  \leq \frac{\epmaximum}{N^2} 
\label{eq:sequence|cover}
\end{align}

Finally, we have:
\begin{align}
\nonumber
& \Pr(\text{neighbor discovery does not finish within } M \text{ stages}) \\
\nonumber & = \Pr(\text{some link is not covered within } M \text{ stages}) \\
\nonumber
& \leq \text{(number of links in the network)} \times \; \frac{\epmaximum}{N^2} 
\leq \epmaximum
\end{align}

Therefore, we have the following theorem:

\begin{theorem}
\Algoref{identical|degree} guarantees that each node discovers all its neighbors on all channels within  $O\left( \frac{\max( \acsmaximum, \dmaximum )}{\srminimum} \log(\destimate)\log\left( \frac{\nsize}{\epmaximum}\right)\right)$ time-slots with probability at least $1 - \epmaximum$.
\end{theorem}

\subsubsection{No Knowledge of Maximum Node Degree}
\label{sec:identical|none}

One way to derive a neighbor discovery algorithm when knowledge about maximum node degree is not available is as follows. Starting with an estimate of one for the maximum node degree, repeatedly run an instance of the algorithm that  assumes knowledge about the maximum degree for a certain number of time-slots with geometrically increasing values for the estimate.  This approach, for example, is used by Vasudevan \emph{et al.} in \cite{VasTow+:2009:MOBICOM} to eliminate the assumption about network size knowledge  made by some of their neighbor discovery algorithms. The approach has the advantage that the time complexity of the knowledge-unaware algorithm is asymptotically same as that of the knowledge-aware algorithm. This approach cannot be used here because it requires computing the exact number of time-slots for which an instance of the knowledge-aware algorithm ought to be run such that, in case the estimate is correct, the neighbor discovery completes with a desired success probability.
Computing the number of time-slots for our (knowledge-aware) algorithm requires nodes to \emph{a priori} know the values for other system parameters, namely $\nsize$, $\acsmaximum$ and $\srminimum$, whose values may not be known in advance.

We instead employ the following approach used in \cite{NakOla:2000:ISAAC}. Starting with an estimate of two, we repeatedly execute an instance of stage with sequentially increasing values for the estimate. 

\begin{algorithm}[t]
\SetKwInput{Input}{Input}
\SetKwInput{Output}{Output}
\SetKwData{DCurrent}{$d$}
\tcp{Algorithm for node $u$}
\BlankLine
\Output{The set of neighbors along with their available channel set}
\BlankLine
\DCurrent $\leftarrow 2$\;
\While{\true}{
execute an instance of stage described in \algoref{identical|degree} with $\destimate$ set to \DCurrent\;
\DCurrent $\leftarrow$ \DCurrent + 1\;
}
\caption{Neighbor discovery algorithm for a synchronous system with identical start times and no knowledge of maximum node degree.}
\label{alg:identical|nodegree}
\end{algorithm}
   
A formal description of the algorithm is shown in \algoref{identical|nodegree}. 
It can be easily verified that, once $d$ becomes at least $\dmaximum$, each stage thereafter contains  a time-slot that satisfies \eqref{eq:slot}. To reach the  success probability of at least $1 - \epmaximum$, from the analysis of \algoref{identical|degree}, it suffices for an execution to contain $M = \frac{16 \max(\acsmaximum, \dmaximum)}{\srminimum} \ln\left( \frac{\nsize^2}{\epmaximum} \right)$ stages each of which consists of a  time-slot that satisfies \eqref{eq:slot}. In other words, neighbor discovery completes successfully with probability at least $1 - \epmaximum$ within $\dmaximum + M$ stages. Note that $M  = \Omega(\dmaximum)$. Therefore we have: 

\begin{theorem}
\Algoref{identical|nodegree} guarantees that each node discovers all its neighbors on all channels within  $O( M \log M )$ time-slots  with probability at least $1 - \epmaximum$, where $M = \frac{16 \max(\acsmaximum, \dmaximum)}{\srminimum}  \ln\left( \frac{\nsize^2}{\epmaximum} \right)$.
\end{theorem}

\subsection{Variable Start Times}

We first assume that the nodes know some upper bound on maximum node degree and present a neighbor discovery algorithm. We then relax this restriction and describe an algorithm for the case when such a knowledge is not available.
Let $\starttime$ denote the time by which all nodes have initiated neighbor discovery.

\subsubsection{Knowledge of Good Upper Bound on Maximum Node Degree}
\label{sec:different|degree}



We assume that nodes know a ``good'' upper bound 
on the maximum node degree. Although the algorithm works even if the upper bound is loose, the running time of the algorithm may be too large since it depends linearly on the value of the upper bound.

The main idea behind our algorithm is to ensure that the transmission probability of a node  is the \emph{same} for every time-slot (but may be different for different nodes).  This allows us to prove that a given link is covered in a time slot with  ``sufficiently high'' probability. Let $\destimate$ denote an upper bound on the maximum node degree as known to all nodes. In each time-slot, a node $u$ randomly selects a channel  from its available channel set, say $c$. It then transmits on $c$ with probability $\min\left( \frac{1}{2}, \frac{\card{\acs{u}}}{\destimate} \right)$ and listens on $c$ with the remaining probability.
A formal description of the algorithm is shown in \algoref{variable|degree}. 


\begin{algorithm}[t]
\SetKwInput{Input}{Input}
\SetKwInput{Output}{Output}
\tcp{Algorithm for node $u$} 
\BlankLine
\Input{An upper bound on maximum node degree, say $\destimate$.}
\Output{The set of neighbors along with the subset of channels that are common with the neighbor.}
\BlankLine 
\BlankLine
\tcp{compute the transmission probability; same for all time-slots}
$p \leftarrow \min\left( \frac{1}{2}, \frac{\card{\acs{u}}}{\destimate} \right)$\;
\While{\true{}}{
\tcp{select a channel}
$c \leftarrow$ channel selected uniformly at random from $\acs{u}$\;
tune the transceiver to $c$\;
switch to transmit mode with probability $p$ and  receive mode with probability $1-p$\;
\uIf{(in transmit mode)}{transmit a message containing $\acs{u}$\;}
\ElseIf{(heard a clear transmission)}{
let the received message be sent by node $v$ containing set $A$\;
add $\ang{v,A \cap \acs{u}}$ to the set of neighbors\;
} 
}
\caption{Neighbor discovery algorithm for a synchronous system with variable start times and knowledge of upper bound on maximum node degree.}
\label{alg:variable|degree}
\end{algorithm}

For the analysis of the running time, as in the case of \algoref{identical|degree}, we can compute the probability of occurrence of events $\A{\tau}{c}$, $\B{\tau}{c}$ and $\C{\tau}{c}$.  We have:
\begin{align}
\nonumber
\Pr\{\A{\tau}{c}\} & = \frac{1}{\card{\acs{v}}} \times \min\left(  \frac{1}{2}, \frac{\card{\acs{v}}}{\destimate}  \right) \\
\nonumber
& = \frac{1}{\max\{ \: 2 \card{\acs{v}}, \destimate \: \}} \\
&  \geq \frac{1}{\max(2 \acsmaximum, \destimate)}
\end{align}
 
It can be verified that the inequalities for $\Pr\{\B{\tau}{c}\}$ in \eqref{eq:B|1} and  $\Pr\{\C{\tau}{c}\}$ in \eqref{eq:C|1} are still valid (although the proofs have to be slightly modified).   Using an analysis similar to that for  \algoref{identical|degree}, we can prove the following result:

\begin{theorem}
\Algoref{variable|degree} guarantees that each node discovers all its neighbors on all channels within  $O\left( \frac{\max( 2\acsmaximum, \destimate )}{\srminimum} \log\left( \frac{\nsize}{\epmaximum}\right)\right)$ time-slots of $\starttime$ with probability at least $1 - \epmaximum$.
\end{theorem}

Note that we no longer have a factor of $O(\log(\destimate))$ in the time complexity because we do not have stages any more.

\subsubsection{No Knowledge of Maximum Node Degree}
\label{sec:different|none}

We refer to the difference in the start times of two nodes as \emph{slack}, and use $\vmaximum$ to denote the maximum slack 
between any two nodes (expressed in terms of the number of time-slots).  
Also, note that no node knows $\vmaximum$. In fact, we present a 
neighbor discovery algorithm in which no node knows any of the other system 
parameters, namely, $\acsmaximum$, $\dmaximum$, $\srminimum$, $\nsize$ and 
$\epmaximum$. And, yet the algorithm completes neighbor discovery within a 
``short'' period with high probability. 

The execution of the algorithm is divided into \emph{epochs}. Each epoch is 
further divided into \emph{phases}. Each phase consists of a certain number 
of time-slots. Specifically, epoch $i$, for $i = 1, 2, 3, \ldots$, consists 
of $i+1$ phases. In phase $j$ of epoch $i$, where $1 \leq j \leq i + 1$, 
we execute an instance of \algoref{variable|degree} with $\destimate$ set to 
$2^j$ for a \emph{fixed} number of time-slots given by 
$2^{i}$. Note that the number of phases in an epoch increases linearily but the length of a 
phase  in an epoch increases geometrically  with the epoch number. Also, note that all 
phases in an epoch are of the same length but use geometrically increasing 
values for estimate of the upper bound on maximum node degree.
A formal description of the algorithm is shown in \algoref{variable|none}. 


\begin{algorithm}[t]
\SetKwInput{Input}{Input}
\SetKwInput{Output}{Output}
\tcp{Algorithm for node $u$} 
\BlankLine
\Output{The set of neighbors along with the subset of channels that are common with the neighbor.}
\BlankLine 
\BlankLine
$i \leftarrow 1$\;
\While{\true{}}{
\tcp{epoch $i$}
\For{$j \leftarrow 1$ \KwTo $i + 1$}{
\tcp{phase $j$ in epoch $i$}
execute an instance of \algoref{variable|degree} with $\destimate = 2^j$ for $2^i$ time-slots\;
}
$i \leftarrow i + 1$\;
}
\caption{Neighbor discovery algorithm for a synchronous system with variable start times (no knowledge assumed).}
\label{alg:variable|none}
\end{algorithm}

In order to show that \algoref{variable|none} guarantees that neighbor discovery completes with high probability within a certain number of time-slots, it is sufficient to identify a phase
in an epoch such that:
\begin{inparaenum}[(i)]
\item nodes use a \emph{sufficiently large value as an estimate} for maximum node degree during this phase, and
\item this phase for all nodes \emph{overlaps for a sufficiently long time} so as to allow all links to be discovered.
\end{inparaenum}

Let $\dmaximum_0$ denote the smallest power of two greater than or equal to 
$\drmaximum$. Further, let $M = 
\frac{16 \max(\acsmaximum, \dmaximum_0)}{\srminimum}  \ln\left( 
\frac{\nsize^2}{\epmaximum} \right)$. Given a node $u$, we refer to the phase  numbered 
$\log \dmaximum_0$ of the epoch numbered $\ceil{\log(M + \vmaximum)}$ as the \emph{\critical{} phase} of node $u$. 
(We call it \critical{} because we show that neighbor discovery is guaranteed to complete with high probability in this phase.)
In its \critical{} phase, each node executes an instance of \algoref{variable|degree} with $\destimate$ set to $\dmaximum_0$ for at least $M + \vmaximum$ time-slots. 
Let $\starttimeOf$ denote the \emph{latest time} at which some node \emph{begins} executing 
its \critical{} phase. Likewise, let $\finishtimeOf$  denote the \emph{earliest time} at 
which some node \emph{ends} executing its \critical{} phase. 
Note that the time difference between when two nodes \emph{start the same phase} is same as the slack between the two nodes (the time difference between when two nodes start the neighbor discovery algorithm). 
Since the maximum slack between any two nodes is 
$\vmaximum$,  clearly, $\finishtimeOf - \starttimeOf \geq M$. 
This implies that, during the time period from $\starttimeOf$ to $\finishtimeOf$, the following holds:
\begin{inparaenum}[(i)]
\item the \critical{} phases of all nodes \emph{mutually overlap} in time, and
\item the time-period contains \emph{at least $M$ time-slots} from the \critical{} phase of every node. 
\end{inparaenum}
We refer to the period of time from $\starttimeOf$ to $\finishtimeOf$ as the \emph{\critical{} period}. 
During the \critical{}  period, each node uses the same estimate for the  upper bound on maximum 
node degree, namely $\dmaximum_0$. Using an analysis similar to that for the previous algorithm, it can be shown that a single time-slot during the \critical{} period 
covers a given link with probability at least $\frac{\srminimum}{16 \max(\acsmaximum, \dmaximum_0)}$. This in turn implies that the probability that neighbor discovery completes successfully  during
the \critical{} period is at least $1 - \epmaximum$.

Note that epoch $i$ contains $(i+1) \cdot 2^i$ time-slots. Thus, the total number of time-slots contained in epochs 1 through $i$ is given by $i \cdot 2^{i+1}$ (proof is given in the appendix). 
As a result, we have:

\begin{theorem}
\Algoref{variable|none} guarantees that each node discovers all its neighbors on all channels within $O((M + \vmaximum)\log( M + \vmaximum ))$  
time-slots of $\starttime$ with probability at least $1 - \epmaximum$, where \linebreak $M = \frac{16 \max(\acsmaximum, \dmaximum_0)}{\srminimum}  \ln\left( 
\frac{\nsize^2}{\epmaximum} \right)$.
\end{theorem}

%% file: asynchronous.tex
\section{Neighbor Discovery in Asynchronous System}
\label{sec:asynchronous}



We first describe an algorithm that assumes that nodes know a good upper bound on maximum node degree. We then describe a  more general algorithm in which no such knowledge is assumed. The second algorithm, however, requires a tighter bound on the clock drift rate than the first algorithm. Both algorithms are derived from their synchronous counterparts in which different nodes can start at different times.

\subsection{Knowledge of Good Upper Bound on Maximum Node Degree}
\label{sec:asynchronous|degree}

Our neighbor discovery algorithm for an asynchronous system is based on our neighbor discovery algorithm for a synchronous system with variable start times (\algoref{variable|degree}).  Let $\destimate$ denote an upper bound on the maximum node degree as known to all nodes. In addition, our algorithm makes the following assumption about the maximum drift rate of the  clock of any node $\drmaximum$: 

\begin{assumption}
\label{asm:drift|rate}
The maximum drift rate of the clock of any node is bounded by $\frac{1}{7}$ seconds/second.
\end{assumption}

The offset or skew  between clocks of any two nodes may be arbitrarily large and, in fact, may grow with time. We now describe how to extend \algoref{variable|degree} to solve the neighbor discovery algorithm in an asynchronous system.

\begin{figure}[tp]
\centerline{\resizebox{4.75in}{!}{\input{frame-local.pstex_t}}}
\caption{\label{fig:frame|local} Execution of a node with respect to its local clock. All frames are of length $\flength$. All slots are of length $\frac{\flength}{3}$.}
\end{figure}
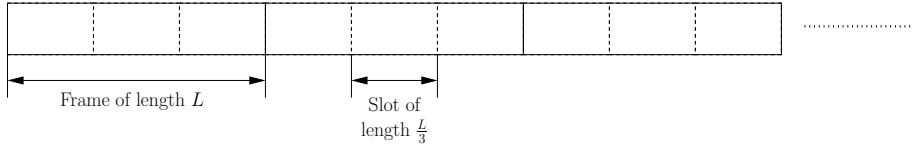

\begin{figure}[tp]
\centerline{\resizebox{5.75in}{!}{\input{frame-global.pstex_t}}}
\caption{\label{fig:frame|global} Execution of the network with respect to real time. Frames may be of different lengths (even within the same node). Slots may be of different lengths (even within the same frame).}
\end{figure}
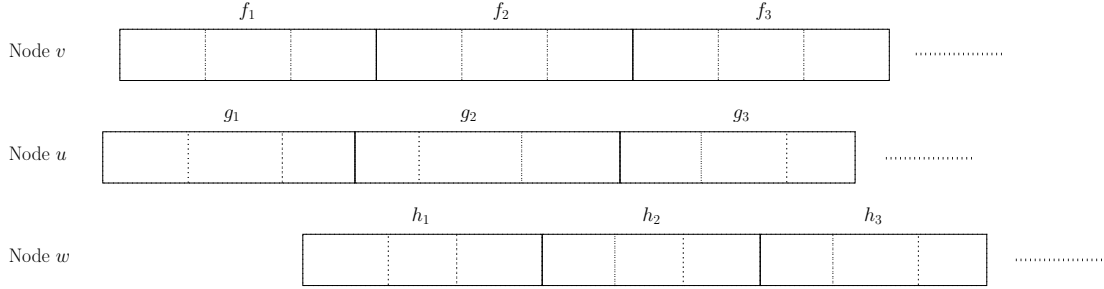

Each node divides  its time into equal-sized \emph{frames}. Frames of different nodes are \emph{not} required to be synchronized and may, in fact, be \emph{misaligned}. A node measures the duration of a frame using its \emph{local} clock. The length of a frame \emph{as measured by a node using its local clock} is same for all nodes, say $\flength$. Note that, because of the clock drift, duration of a frame, when projected on real time, may be different from $\flength$---shorter than $\flength$ in case of positive drift and longer than $L$ in case of negative drift. Specifically, it can be verified that the length of a frame in real time lies in the range:
\begin{equation}
\label{eq:range|realtime}
\frac{\flength}{1 + \drmaximum} \; \leq \; \text{frame length in real-time} \; \leq \; \frac{\flength}{1 - \drmaximum}
\end{equation}
A node further divides each frame into three equal-sized slots---equal with respect to its local clock. Therefore the duration of a slot as measured by a node using its local clock is $\frac{\flength}{3}$. At the beginning of each frame, a node $u$ randomly selects a channel from its available channel set, say $c$. It then transmits on $c$ during each slot of the frame with probability $\min\left( \frac{1}{2}, \frac{\card{\acs{u}}}{3 \destimate}\right)$ and listens on $c$ during the entire frame with the remaining probability. In the former case, $u$ transmits the same message during each slot of the frame. In the latter case, $u$ listens for any clear messages it may receive during \emph{any part} of the frame. The partition of a frame into slots is not important for a receiving node; they are only used by a transmitting node. Also, a node may receive multiple clear messages while listening during a single frame.  The execution of a node with respect to its local clock is shown in \figref{frame|local}. The execution of the network with respect to common real time is shown in \figref{frame|global}.

\begin{algorithm}[t]
\SetKwInput{Input}{Input}
\SetKwInput{Output}{Output}
\tcp{Algorithm for node $u$} 
\BlankLine
\Input{An upper bound on maximum node degree, say $\destimate$.}
\Output{The set of neighbors along with the subset of channels that are common with the neighbor.}
\BlankLine 
\BlankLine
\tcp{compute the transmission probability}
$p \leftarrow \min\left( \frac{1}{2}, \frac{\card{\acs{u}}}{3\destimate} \right)$\;
\While{\true{}}{
\tcp{select a channel}
$c \leftarrow$ channel selected uniformly at random from $\acs{u}$\;
tune the transceiver to $c$\;
switch to transmit mode with probability $p$ and  receive mode with probability $1-p$\;
\uIf{(in transmit mode)}{transmit a message containing $\acs{u}$ during each slot of the frame\;}
\Else{
\tcp{in receive mode}
\ForEach{(clear message received during the frame)}{
let the received message be sent by node $v$ containing set $A$\;
add $\ang{v,A \cap \acs{u}}$ to the set of neighbors\;
}
} 
}
\caption{Neighbor discovery algorithm for an asynchronous system with a bound on clock drift rate and knowledge of upper bound on maximum node degree.}
\label{alg:asynchronous|degree}
\end{algorithm}

A formal description of the algorithm is shown in \algoref{asynchronous|degree}. We now analyze the running time of the algorithm. In our analysis, unless otherwise stated, time refers to \emph{real-time}. Of course, nodes do not have access to the real-time. 
We define two notions we use in our analysis as follows:

\begin{definition}[\legal{} pair] We say that a pair of frames $\lpair{f}{g}$ is \emph{\legal}  if at least one slot of $f$ lies completely within $g$. 
\end{definition}

\begin{definition}[overlapping frames]
For a  frame $f$ and a node $u$, let $\overlap{f}{u}$ denote the subset of frames of $u$ that overlap in real-time with $f$.  Further, let $\overlapAll{f}$ denote the subset of all frames that overlap in real-time with $f$.
\end{definition}

For example, as per the execution shown in \figref{frame|global}, pairs $\lpair{f_1}{g_1}$ and $\lpair{f_2}{h_1}$ are \legal{} whereas the pair $\lpair{f_1}{h_1}$ is not. Also, $\overlap{g_2}{v} = \{ f_1, f_2 \}$ and $\overlapAll{g_2} = \{ f_1, f_2,  g_2, h_1, h_2 \}$. 

One of the arguments we commonly use in our analysis is:  ``$x$ adjacent slots of a node cannot strictly contain $y$ adjacent slots of another node'' for certain specific values of $x$ and $y$ with $x < y$. To prove this, we argue that, for the statement to be false, the following inequality must hold:
\[
\frac{xL}{3 (1 - \drmaximum)} > \frac{y L}{3(1+ \drmaximum)} 
\]
and then show that it contradicts \asmref{drift|rate}. The above inequality must hold for the statement to be false because the left hand side denotes the largest possible length of the time interval containing $x$ adjacent slots, and the right hand side denotes the smallest possible length of the time interval containing $y$ adjacent slots. 

\begin{lemma} 
\label{lem:overlap}
A frame of a node overlaps with at most three frames of any other node. Formally,
\[\forall f, u ::  \card{\overlap{f}{u}} \leq 3\]
\end{lemma}
\begin{proof}
Assume, in order to obtain a contradiction, that a frame of some node, say $f$, overlaps with at least four frames of another node. This implies that there are at least two consecutive frames that are strictly contained within $f$. From \eqref{eq:range|realtime}, a frame of a node can strictly contain two frames of another node only if the following condition holds:
\[
\frac{\flength}{1 - \drmaximum} > \frac{2\flength}{1 + \drmaximum}  \quad \implies \quad \drmaximum > \frac{1}{3}
\]
This contradicts \asmref{drift|rate}.  
\end{proof}

Consider the link from node $v$ to node $u$ on channel $c$. Also, consider frames $f$ of $v$ and $g$ of $u$ such that the pair $\lpair{f}{g}$ is \legal{}. We extend the notion of a link covered by a time-slot (defined in \secref{identical|degree}) to a link covered by an \legal{} pair of frames. Specifically, the pair of \legal{} frames $\lpair{f}{g}$ \emph{covers} the link $\link{v}{u}$ on channel $c$ if: \begin{inparaenum}[(i)]
\item $v$ transmits on $c$ during $f$, 
\label{enum:pair|1}
\item $u$ listens on $c$ during $g$, and 
\label{enum:pair|2} 
\item no other neighbor of $u$, say $w$, transmits on $c$ during any frame in $\overlap{g}{w}$.
\label{enum:pair|3}  
\end{inparaenum} Also, $\lpair{f}{g}$ covers $\link{v}{u}$ if $\lpair{f}{g}$ covers $\link{v}{u}$ on some channel $c \in \lspan{v}{u}$.
The three conditions  collectively ensure that $u$ receives a clear message from $v$ (on $c$ during $g$) provided $\lpair{f}{g}$ is \legal{}. We have:

\begin{lemma}
\label{lem:legal|atleast}
If $\lpair{f}{g}$ is \legal{}, then $\lpair{f}{g}$ covers $\link{v}{u}$ with probability at least $\frac{\srminimum}{8\max(2\acsmaximum, 3\destimate)}$.
\end{lemma}

\begin{proof}
Analogous to $\A{\tau}{c}$, $\B{\tau}{c}$ and $\C{\tau}{c}$ defined in \secref{identical|degree}, let $\AF{f}{g}{c}$, $\BF{f}{g}{c}$, $\CF{f}{g}{c}$ denote the three events corresponding to the three conditions \enumref{pair|1}, \enumref{pair|2} and \enumref{pair|3}, respectively, required for coverage.
Note that the three events are \emph{mutually independent}. As before, we compute the  probability of occurrence of the three events separately.

\paragraph{Computing the probability of occurrence of $\AF{f}{g}{c}$} We have:
\begin{align}
\nonumber
\Pr\{\AF{f}{g}{c}\} & =  (v \text{ selects } c \text{ at the beginning of } f) \land \\
\nonumber
& \phantom{=~} (v \text{ chooses to transmit during } f) \\
\nonumber
& =   \frac{1}{\card{\acs{v}}} \times \min\left( \frac{1}{2}, \frac{\card{\acs{v}}}{3 \destimate}  \right) \\
\nonumber
& =  \min\left( \frac{1}{2\card{\acs{v}}}, \frac{1}{3 \destimate} \right) \\
& \geq \frac{1}{\max( 2 \acsmaximum, 3 \destimate )} 
\label{eq:AA|1}
\end{align}

\paragraph{Computing the probability of occurrence of  $\BF{f}{g}{c}$} We have:
\begin{align}
\nonumber
\Pr\{\BF{f}{g}{c}\} & =  (u \text{ selects } c \text{ at the beginning of } g) \land \\
\nonumber
& \phantom{=~} (u \text{ chooses to listen during } g) \\
\nonumber
& =  \frac{1}{\card{\acs{u}}} \times \left\{1 - \min\left( \frac{1}{2}, \frac{\card{\acs{u}}}{3 \destimate}  \right) \right\} \\
\nonumber
& \phantom{=} \big\{ \min(x, y) \leq x \big\} \\
& \geq  \frac{1}{\card{\acs{u}}} \times  \left( 1 - \frac{1}{2} \right)  \; = \; \frac{1}{2 \card{\acs{u}}}
\label{eq:BB|1}
\end{align}

\paragraph{Computing the probability of  occurrence of $\CF{f}{g}{c}$} 
Let $\incoming{u}{c}$ denote the set of  neighbors of $u$ on $c$. Note that, if $\incoming{u}{c}$ only contains $v$, then $\Pr\{ \CF{f}{g}{c} \} = 1$. Otherwise, we have:
\begin{align}
\nonumber
& \Pr\{\CF{f}{g}{c}\} \\
\nonumber
& =  \prod_{\substack{w \in \incoming{u}{c} \\ w \neq v}} \!\! \Pr\left( \begin{array}{@{}c@{}} w \text{ does not transmit on } c \text{ during any frame} \\[-0.125em] \text{in } \overlap{g}{w} \end{array} \right) \\
\nonumber
& =  \prod_{\substack{w \in \incoming{u}{c} \\ w \neq v}} \prod_{h \in \overlap{g}{w}} \!\!
\Pr\left( \begin{array}{@{}c@{}} w \text{ does not transmit on } c \\[-0.125em] \text{ during  } h \end{array} \right)
 \\
\nonumber
& =  \prod_{\substack{w \in \incoming{u}{c} \\ w \neq v}}  \prod_{h \in \overlap{g}{w}} \!\! \{1 - \Pr(w \text{ transmits on } c \text{ during } h)\} \\
\nonumber
& \geq  \prod_{\substack{w \in \incoming{u}{c} \\ w \neq v}}  \prod_{h \in \overlap{g}{w}} \!\! \left\{ 1 -  \frac{1}{\card{\acs{w}}} \times \min\left( \frac{1}{2}, \frac{\card{\acs{w}}}{3 \destimate} \right) \right\} \\
\nonumber
& =  \prod_{\substack{w \in \incoming{u}{c} \\ w \neq v}}  \prod_{h \in \overlap{g}{w}} \!\! \left\{ 1 - \min\left( \frac{1}{2\card{\acs{w}}}, \frac{1}{3 \destimate} \right) \right\} \\
\nonumber
& \phantom{=} \big\{ \min(x,y) \leq y \big\} \\
\nonumber
& \geq  \prod_{\substack{w \in \incoming{u}{c} \\ w \neq v}} \prod_{h \in \overlap{g}{w}} \!\!  \left( 1 - \frac{1}{3 \destimate} \right)  \\
\nonumber
& \phantom{=} \big\{ \text{using \lemref{overlap}, } \card{\overlap{g}{w}} \leq 3 \big\} \\
\nonumber
& \geq  \prod_{\substack{w \in \incoming{u}{c} \\ w \neq v}}  \left( 1 - \frac{1}{3 \destimate} \right)^{3} \\ 
\nonumber 
& =  \left(1 - \frac{1}{3 \destimate} \right)^{3 (\card{\incoming{u}{c}} - 1)} \\
\nonumber 
& \phantom{=} \big\{  \card{\incoming{u}{c}} - 1 =  \degreeon{u}{c} -1 \leq \destimate \big\}  \\
\nonumber
& \geq  \left(1 - \frac{1}{3 \destimate} \right)^{3 \destimate}  \\
\nonumber
& \phantom{=} \left\{ \forall x \geq 2,  \left( 1 - \tfrac{1}{x} \right)^x \text{ is a monotonically increasing} \right. \\
\nonumber
& \phantom{=\left\{\right.} \left.\text{function of } x \text{ and thus } \geq \tfrac{1}{4} \right\}  \\
& \geq  \frac{1}{4} 
\label{eq:CC|1}
\end{align}

Finally, using a derivation similar to that for \eqref{eq:stage|cover}, it can be  shown that $\lpair{f}{g}$ covers $\link{v}{u}$ with probability at least $\frac{\srminimum}{8\max(2\acsmaximum, 3\destimate)}$. 
\end{proof}

For a frame $f$, we use $\node{f}$ to denote the node to which $f$ belongs. For example, as per the execution show in \figref{frame|global}, $\node{f_1} = v$ and $\node{h_1} = w$. We now define a precedence relation between pairs of frames referred to as  \emph{frame-pairs} for short.
 
\begin{definition}[\precedence{} relation]
For frame-pairs $\lpair{f}{g}$ and $\lpair{p}{q}$, we say that $\lpair{f}{g}$ \emph{precedes} $\lpair{p}{q}$, denoted $\lpair{f}{g} \under{} \lpair{p}{q}$, if \begin{inparaenum}[(i)]  \item $\node{f} = \node{p}$, \item $\node{g} = \node{q}$, \item the start time of $f$ is before that of  $p$, and \item the start time of $g$ is before that of  $q$. 
\end{inparaenum}
\end{definition}

For example, as per the execution shown in \figref{frame|global}, $\lpair{f_1}{g_1} \under \lpair{f_2}{g_2}$ and $\lpair{f_1}{g_1} \under \lpair{f_3}{g_2}$ but $\lpair{f_1}{g_1} \not\under \lpair{f_1}{g_2}$. We next define a notion on a sequence of frame-pairs that intuitively enables us to treat coverage provided by different frame-pairs as essentially independent events.

\begin{definition}[\admissible{} sequence] A sequence $\sigma$ of $M$ frame-pairs $\big\{ \lpair{f_{i_k}}{g_{j_k}} \big\}_{1 \leq k \leq M}$ is \emph{\admissible} with respect to the link $\link{v}{u}$ if it satisfies the following conditions:
\begin{enumerate}
\item $\forall i: 1 \leq k \leq M: \node{f_{i_k}} = v$ and $\node{g_{j_k}} = u$,
\item $\forall i: 1 \leq k < M: \lpair{f_{i_k}}{g_{j_k}} \under \lpair{f_{i_{k+1}}}{g_{j_{k+1}}}$,
\item $\forall i: 1 \leq k \leq M: \lpair{f_{i_k}}{g_{j_k}}$ is \legal{}, and
\item $\forall i: 1 \leq k < M: \overlapAll{g_{i_k}} \cap \overlapAll{g_{i_{k+1}}}$ $= \emptyset$.
\end{enumerate}
\end{definition}

For a sequence $\sigma$ of frame-pairs that is admissible with respect to $\link{v}{u}$, we say that $\sigma$ covers $\link{v}{u}$ if some frame-pair in the sequence covers $\link{v}{u}$.

\begin{lemma}
\label{lem:sequence|covers}
Let $\sigma$ be a sequence of $\frac{8\max(2\acsmaximum, 3 \destimate)}{\srminimum} \ln\left( \frac{\nsize^2}{\epsilon} \right)$ frame-pairs such that $\sigma$ is admissible with respect to $\link{v}{u}$. Then the probability that $\sigma$ does not cover $\link{v}{u}$ is at most  $\frac{\epmaximum}{\nsize^2}$.
\end{lemma}
\begin{proof} For convenience, let $M = \frac{8\max(2\acsmaximum, 3 \destimate)}{\srminimum} \ln\left( \frac{\nsize^2}{\epsilon} \right)$. Let $E_k$ with $1 \leq k \leq M$ denote the event that $\lpair{f_{i_k}}{g_{j_k}}$ covers $\link{v}{u}$ on some channel. Note that the occurrence of the event $E_k$ depends \emph{only} on frames in $\overlapAll{g_{j_k}}$. Since $\sigma$ is an admissible sequence, for all $x$ and $y$, with $1 \leq x < y \leq M$, $\overlapAll{g_{j_x}} \cap \overlapAll{g_{j_y}} = \emptyset$. In other words, the set of frames that overlap with the frame $g_{j_x}$ are 
distinct from the set of frames that intersect the frame $g_{j_y}$. Since, for each frame, a node chooses its action randomly, events $E_k$s are \emph{mutually independent} of each other. Using a derivation similar to that of \eqref{eq:sequence|cover}, it can be shown that $\sigma$ does not cover $\link{v}{u}$ with probability at most $\frac{\epmaximum}{\nsize^2}$.
\end{proof}

We next show that any execution of the network must contain a ``sufficiently long'' sequence of admissible frame-pairs. 
In the remainder of this section, let $\starttime$ denote the time by which all nodes have initiated the neighbor discovery algorithm, \algoref{asynchronous|degree}.
To show the existence of a ```sufficiently long'' admissible sequence, we first show that, for any instant of time $T$ after $\starttime$ and, for any pair of neighboring nodes, there is an \legal{} 
pair of frames after but ``close'' to $T$.

\begin{lemma}
\label{lem:legal|close}
Consider a link from node $v$ to node $u$ and some instant of time $T$ with $T \geq \starttime$. Let $f_i$ (respectively, $g_i$) with $i \geq 1$ denote the $i^{th}$ full frame of node $v$ (respectively, node $u$)  \emph{after} $T$. Then some frame in $\{ f_1, f_2 \}$ is \legal{} with some frame in $\{ g_1, g_2 \}$.
\end{lemma}
\begin{proof}
Note that the frames $f_1$ and $f_2$ together contain six slots. Let $\slotv{i}$ for $i = 1, 2, \ldots, 6$ denote the start times of the six slots numbered in the increasing order of their start times (see \figref{slots}). Likewise, let $\slotu{i}$ for $i = 1, 2, \ldots, 6$ denote the start times of the six slots of the frames $g_1$ and $g_2$. \\

\begin{figure}[tp]
\centerline{\resizebox{\columnwidth}{!}{\input{slots.pstex_t}}}
\caption{\label{fig:slots} Slots used in the proof of \lemref{legal|close}.}
\end{figure}

\begin{figure}[tp]
\centerline{\resizebox{\columnwidth}{!}{\input{cases.pstex_t}}}
\caption{\label{fig:cases} Various cases in the proof of \lemref{legal|close}.}
\end{figure}

\noindent
\textit{Claim 1:} 
We first show that $\slotu{1} \leq \slotv{5}$.  Assume, by the way of contradiction, that $\slotv{5} < \slotu{1}$. Note that, by definition of $g_1$, there is only a partial frame of $u$ between $T$ and $\slotu{1}$. This, in turn, implies that a partial frame of $u$ contains at least four slots of $v$. This can happen only if the following condition holds:
\[
\frac{\flength}{1 - \drmaximum} > \frac{4\flength}{3 (1 + \drmaximum)} \quad \implies \quad \drmaximum > \frac{1}{7} 
\]
This contradicts \asmref{drift|rate}. \\

\noindent
\textit{Claim 2:} 
Likewise, we can show that $\slotv{1} \leq \slotu{5}$. \\

\noindent
\textit{Claim 3:} 
If $\slotv{i-1} \leq \slotu{1} \leq \slotv{i}$ for some $i$ with $1 < i \leq 5$, then  $\slotu{1} \leq \slotv{i} < \slotv{i+1} \leq \slotu{4}$. It suffices to show that $\slotv{i+1} \leq \slotu{4}$. Assume, by the way of contradiction, that $\slotu{4} < \slotv{i+1}$. This implies that $\slotv{i-1} \leq \slotu{1} < \slotu{4} < \slotv{i+1}$. In other words, two adjacent slots of $v$ strictly contain an entire frame of $u$. This can happen only if the following condition holds:
\[
\frac{2 \flength}{3(1 - \drmaximum)} > \frac{\flength}{1 + \drmaximum} \quad \implies \quad \drmaximum > \frac{1}{5}
\]
This contradicts \asmref{drift|rate}. \\

We now prove the lemma statement. 
We consider two cases depending on whether $f_1$ or $g_1$ has earlier start time.

\begin{itemize}
\item \textit{Case 1 ($\slotv{1} < \slotu{1}$):} In this case, we have $\slotv{1} <  \slotu{1} \leq \slotv{5}$ (see \figref{cases}(a)). The rightmost inequality follows from Claim~1. The time interval from $\slotv{1}$ to $\slotv{5}$ consists of four contiguous slots of $v$: $[\slotv{1}, \slotv{2}]$, $[\slotv{2}, \slotv{3}]$, $[\slotv{3}, \slotv{4}]$ and $[\slotv{4}, \slotv{5}]$. Clearly,  $\slotu{1}$ lies in at least one of them. (If $\slotu{1}$ lies on the boundary of two slots, we select the earlier one.) In any case, using Claim~3, we can show that at least  one of the slots $[\slotv{2}, \slotv{3}]$, $[\slotv{3}, \slotv{4}]$, $[\slotv{4}, \slotv{5}]$ or $[\slotv{5}, \slotv{6}]$ is contained in the frame $g_2$. This implies that either $\lpair{f_1}{g_1}$ or $\lpair{f_2}{g_1}$ is \legal{}.

\item \textit{Case 2 ($\slotu{1} \leq \slotv{1}$):} If $\slotv{2} \leq \slotu{4}$, then $\slotu{1} \leq \slotv{1} < \slotv{2} \leq \slotu{4}$ (see \figref{cases}(b)). In other words, the first slot of the frame $f_1$ is contained with the frame $g_1$, which implies that $\lpair{f_1}{g_1}$ is \legal{}. Therefore assume that $\slotu{4} < \slotv{2}$. We have two subcases depending on where $\slotv{1}$ lies relative to $\slotu{4}$. Let $\slotu{7}$ denote the end time of the frame $g_2$.
\begin{itemize}
\item \textit{Case 2.1 ($\slotv{1} \leq \slotu{4}$):}  In this case, we show that the slot $[\slotv{2}, \slotv{3}]$ is contained within the frame $[\slotu{4}, \slotu{7}]$ (see \figref{cases}(c)).  To that end, we first prove that $\slotv{3} < \slotu{7}$. If not, then $\slotv{1} \leq \slotu{4} < \slotu{7} \leq \slotv{3}$. In other words, two adjacent slots of $v$ strictly contain an entire frame of $u$. This can happen only if the following condition holds:
\[
\frac{2 \flength}{3(1 - \drmaximum)} > \frac{\flength}{1 + \drmaximum} \quad \implies \quad \drmaximum > \frac{1}{5}
\]
This contradicts \asmref{drift|rate}. Therefore, we have $\slotu{4} < \slotv{2} < \slotv{3} < \slotu{7}$, which implies that the pair $\lpair{f_1}{g_2}$ is \legal{}.
\item \textit{Case 2.2 ($\slotu{4} < \slotv{1}$):} In this case, we show that the slot $[\slotv{1}, \slotv{2}]$ is contained within the frame $[\slotu{4}, \slotu{7}]$ (see \figref{cases}(d)). To that end, we show that $\slotv{2} \leq \slotu{7}$. If not, then $\slotv{1} \leq \slotu{5} < \slotu{7} < \slotv{2}$. (The inequality $\slotv{1} \leq \slotu{5}$ follows from Claim~2.) In other words, one slot of $v$ strictly contains at least two adjacent slots of $u$. This can happen only if the following condition holds:
\[
\frac{\flength}{3(1 - \drmaximum)} > \frac{2 \flength}{3(1 + \drmaximum)} \quad \implies \quad \drmaximum > \frac{1}{3}
\]
This contradicts \asmref{drift|rate}. Therefore, we have $\slotu{4} < \slotv{1} < \slotv{2} \leq \slotu{7}$, which implies that the pair $\lpair{f_1}{g_2}$ is \legal{}.
\end{itemize}
\end{itemize}

In all cases, we show that one of the four pairs $\lpair{f_1}{g_1}$, $\lpair{f_1}{g_2}$, $\lpair{f_2}{g_1}$ or $\lpair{f_2}{g_2}$ is aligned.
\end{proof}

We now show that existence of a ``sufficiently long'' admissible sequence in any execution of \algoref{asynchronous|degree}.

\begin{lemma}
\label{lem:subsequence}
Consider a link from node $v$ to node $u$.  Further, consider an execution of the network after $\starttime$ that contains at least $M$ full frames of $u$ as well as $v$. Then the execution contains a sequence of at least $\frac{M}{6}$ frame-pairs such that the sequence is admissible with respect to $\link{v}{u}$. 
\end{lemma}
\begin{proof}
The proof is by construction. The construction is in two steps. In the first step, we construct a sequence of frame-pairs $\gamma$ that is ``almost'' admissible in the sense that it satisfies the first three properties of an admissible sequence but may not satisfy the fourth property. We show that $\gamma$ contains at least $\frac{M}{2}$ frame-pairs. In the second step, using $\gamma$, we construct a sequence of frame-pairs $\sigma$ that satisfies all four properties of
an admissible sequence. We also show that $\sigma$ contains at least $\frac{M}{6}$ frame-pairs.

\paragraph*{Constructing $\gamma$} To obtain the first frame-pair $\lpair{f_{i_1}}{g_{i_1}}$ that is \legal{}, we apply \lemref{legal|close} to $\starttime$. Now, assume that we have already selected $k$ frame-pairs satisfying the first three properties of an admissible sequence. Let the $k^{th}$ frame-pair be denoted by $\lpair{f_{i_k}}{g_{i_k}}$. To select the next frame-pair, let $T_k$ be defined as the \emph{earlier} of the end times of frames $f_{i_k}$ and $g_{i_k}$. 
To obtain the next frame-pair that is \legal{}, we apply \lemref{legal|close} to $T_k$. Let the frame-pair be denoted by  $\lpair{f_{i_{k+1}}}{g_{i_{k+1}}}$. Clearly, the extended sequence (consisting of $k+1$ \legal{} frame-pairs) satisfies the first and third properties of an admissible sequence. The second property holds because the start times of both $f_{i_k}$ and $g_{i_k}$ are \emph{before} $T_k$, whereas the start times of both $f_{i_{k+1}}$ and $g_{i_{k+1}}$ are \emph{after} $T_k$. We repeatedly select \legal{} frame-pairs using \lemref{legal|close} until we run out of frames of either $u$ or $v$. We now establish a lower bound on the length of $\gamma$. Note that, when selecting the $(k+1)^{st}$ pair, the first full frame of node $v$ after $T_k$ (namely frame $f_1$ in the statement of \lemref{legal|close}) is \emph{adjacent} to the frame $f_{i_k}$. This follows from the definition of $T_k$. This, in turn, implies that the frame $f_{i_{k+1}}$ obtained using \lemref{legal|close} is within a distance of two of the frame $f_{i_k}$.  Likewise, the frame $g_{i_{k+1}}$ is within a distance of two of the frame $g_{i_k}$. As a result, $\gamma$ contains at least $\frac{M}{2}$ frame-pairs.

\paragraph*{Constructing $\sigma$} To construct a sequence $\sigma$ that also satisfies the fourth property of an admissible sequence, we choose every \emph{third} frame-pair of $\gamma$ starting with the first frame-pair $\lpair{f_{i_1}}{g_{i_1}}$. Clearly $\sigma$ also satisfies the first three properties of an admissible sequence (since it is a subsequence of $\gamma$). Let the $k^{th}$ frame-pair of $\sigma$ be denoted by $\lpair{f_{j_k}}{g_{j_k}}$. To prove that $\sigma$ satisfies the fourth property as well, assume, by the way of contradiction, that some frame, say $h$, overlaps with two consecutive frames $g_{j_k}$ and $g_{j_{k+1}}$ for some $k$. Note that, since we selected only every third frame-pair of $\gamma$ to construct $\sigma$, there are at least two other frames of $u$ between $g_{j_k}$ and $g_{j_{k+1}}$. This implies that $h$ overlaps with at least four frames of $u$, which contradicts \lemref{overlap}. This establishes that the length of $\sigma$ is at least $\frac{M}{6}$. 
\end{proof}

Finally, we have the main result.

\begin{theorem}
\label{thm:asynchronous|main}
Let $\starttime$ be the time by which all nodes have initiated neighbor discovery. Also, let $\finishtime$ be the \emph{earliest} time by which each node has executed at least $\frac{48\max(2\acsmaximum, 3 \destimate)}{\srminimum} \ln\left( \frac{\nsize^2}{\epsilon} \right)$  full frames since $\starttime$. Then \algoref{asynchronous|degree} ensures that each node discovers all its neighbors  on all channels with probability at least $1 - \epmaximum$ by time $\finishtime$. 
\end{theorem}
\begin{proof}
Consider a link $\link{v}{u}$. 
From  \lemref{subsequence},  the execution from $\starttime$ to $\finishtime$ contains a sequence of at least  $\frac{8\max(2\acsmaximum, 3 \destimate)}{\srminimum} \ln\left( \frac{\nsize^2}{\epsilon} \right)$  frame-pairs such that the sequence is admissible with respect to $\link{v}{u}$. From \lemref{sequence|covers}, the probability that $\link{v}{u}$ is not covered by $\finishtime$ is at most $\frac{\epmaximum}{\nsize^2}$. This implies that the probability that some link in the network is not covered by $\finishtime$ is at most $\epmaximum$. 
\end{proof}

We now bound the length of the interval $\finishtime - \starttime$.

\begin{theorem} 
\label{thm:asynchronous|time}
Let $\starttime$ and $\finishtime$ be as defined in \thmref{asynchronous|main}. Then the length of the interval 
$\finishtime - \starttime$ is upper bounded by $\left\{ \frac{48 \max(2\acsmaximum, 3 \destimate)}{\srminimum} \ln\left( \frac{\nsize^2}{\epsilon} \right) + 1 \right\} \left( \frac{\flength}{1 - \drmaximum} \right)$.
\end{theorem}
\begin{proof}
By our choice of $\finishtime$, there exists some node such that $\finishtime - \starttime$ contains exactly \linebreak $\frac{48 \max(2\acsmaximum, 3 \destimate)}{\srminimum} \ln\left( \frac{\nsize^2}{\epsilon} \right)$ full frames of that node. The first frame of that node in the execution may be partial frame. From \eqref{eq:range|realtime}, the length of each frame is upper bounded by $\frac{\flength}{1 - \drmaximum}$. Combining the three facts, we obtain the result.
\end{proof}

In our algorithm, we divided our frame into three slots. It is possible to divide a frame into only two slots and a result similar to \thmref{asynchronous|time} can be derived. However, the proof is more involved. Likewise, a frame can be divided into four or more slots albeit at the expense of increasing the running time of the algorithm.
 
\subsection{No Knowledge of Maximum Node Degree}
\label{sec:asynchronous|none}

The algorithm for this case is similar to the one for a synchronous system, namely \algoref{variable|none}, but with time-slots replaced with frames and transmission probability calculated as in \algoref{asynchronous|degree}. A formal description of the algorithm is given in \algoref{asynchronous|none}.

\begin{algorithm}[t]
\SetKwInput{Input}{Input}
\SetKwInput{Output}{Output}
\tcp{Algorithm for node $u$} 
\BlankLine
\Output{The set of neighbors along with the subset of channels that are common with the neighbor.}
\BlankLine 
\BlankLine
$i \leftarrow 1$\;
\While{\true{}}{
\tcp{epoch $i$}
\For{$j \leftarrow 1$ \KwTo $i + 1$}{
\tcp{compute the transmission probability}
$p \leftarrow \min\left( \frac{1}{2}, \frac{\card{\acs{u}}}{3 \cdot 2^j} \right)$\;
\For{$k \leftarrow 1$ \KwTo $2^i$}{
\tcp{phase $j$ in epoch $i$}
\tcp{select a channel}
$c \leftarrow$ channel selected uniformly at random from $\acs{u}$\;
tune the transceiver to $c$\;
switch to transmit mode with probability $p$ and  receive mode with probability $1-p$\;
\uIf{(in transmit mode)}{transmit a message containing $\acs{u}$ during each slot of the frame\;}
\Else{
\tcp{in receive mode}
\ForEach{(clear message received during the frame)}{
let the received message be sent by node $v$ containing set $A$\;
add $\ang{v,A \cap \acs{u}}$ to the set of neighbors\;
}
}
}
}
$i \leftarrow i + 1$\;
}
\caption{Neighbor discovery algorithm for an asynchronous system with a bound on clock drift rate.}
\label{alg:asynchronous|none}
\end{algorithm}

The time-complexity analysis of \algoref{asynchronous|none} is more involved because of asynchrony. Unlike in a synchronous system, the time difference between when two nodes start the same phase may \emph{grow} with time due to clock drift. We need some additional notation for the analysis.

Let $\phase{i}{j}$ denote the phase $j$ of the epoch $i$. Let $\framesuntil{i}{j}$ denote the number of frames that a node in the execution until (but not including) $\phase{i}{j}$. We have:
\begin{align}
\nonumber
\framesuntil{i}{j} 
& = \text{(number of frames contained in the first } i-1 \text{ epochs)} \quad + \\
\nonumber
& \phantom{=} \text{(number of frames contained in the } i^\text{th} \text{ epoch until } \phase{i}{j})  \\
\nonumber
& \phantom{=} \left\{ \text{using \eqref{eq:epoch|sum} in the appendix and the fact that, for each } j, \card{\phase{i}{j}} = 2^i  \right\} \\
\nonumber
& = (i-1) \cdot 2^i + (j-1) \cdot 2^i \\
& = (i + j - 2) \cdot 2^i
\label{eq:frames|until}
\end{align}

Let $\starttimeOf{i}{j}{u}$ denote the time at which node $u$ begins executing $\phase{i}{j}$. Likewise, let $\finishtimeOf{i}{j}{u}$ denote the time at which node $u$ ends executing $\phase{i}{j}$.  Let $\starttimeOf{i}{j}$ denote the \emph{latest} time at which some node starts executing $\phase{i}{j}$. Likewise, let  $\finishtimeOf{i}{j}$ denote the \emph{earliest} time at which some node starts executing $\phase{i}{j}$.
Symbols $\starttime$, $\vmaximum$ and $\flength$ have the same meaning as before. Note that, in real-time, the minimum and maximum lengths of a frame are given by $\frac{\flength}{1 + \drmaximum}$ and $\frac{\flength}{1 - \drmaximum}$, respectively. Thus, we have:
\begin{alignat}{4}
& \starttimeOf{i}{j} \quad  && = \quad \max \{ \starttimeOf{i}{j}{u} \}  \quad && \leq \quad \starttime{} && + \framesuntil{i}{j} \left( \frac{\flength}{1 - \drmaximum} \right) 
\label{eq:phase|min} \\
& \finishtimeOf{i}{j} \quad && = \quad \min \{ \finishtimeOf{i}{j}{u} \} \quad && \geq \quad (\starttime{} - \vmaximum) && + (\framesuntil{i}{j} + 2^i) \left( \frac{\flength}{1 + \drmaximum} \right)
\label{eq:phase|max}
\end{alignat}

Suppose $i$ and $j$ are such that $\finishtimeOf{i}{j} < \starttimeOf{i}{j}$. 
Note that, during the time-period from $\starttimeOf{i}{j}$ to $\finishtimeOf{i}{j}$, all nodes will use the \emph{same} estimate for maximum node degree.      
Recall that $\dmaximum_0$ denotes the smallest power of two greater than or equal to 
$\drmaximum$.
Let $M =  \frac{48 \max(2\acsmaximum, 3 \dmaximum_0)}{\srminimum} \ln\left( \frac{\nsize^2}{\epsilon} \right) + 1$. From \thmref{asynchronous|main}, if the time period from $\starttimeOf{i}{j}$ to $\finishtimeOf{i}{j}$ is such that:
\begin{inparaenum}[(a)] 
\item the time period contains at least $M$ full frames and 
\item each node uses an estimate value of $\dmaximum_0$ during the time period, 
\end{inparaenum} then the neighbor discovery is guaranteed to complete with probability at least $1 - \epmaximum$. 
Let $i_0 = \left\lceil \log \left( 6 M + \frac{3\vmaximum}{\flength} \right) \right\rceil$ and $j_0 = \ceil{ \log \dmaximum}  (= \log \dmaximum_0)$. Note that, by definition, during $\phase{i_0}{j_0}$, each node uses
$\dmaximum_0$ as the estimate for maximum node degree. It suffices to prove that the time period from $\starttime{i_0}{j_0}$ to $\finishtime{i_0}{j_0}$ contains at least $M$ frames. To that end, we make 
the following assumption about the clock drift rate:

\begin{assumption}
\label{asm:drift|rate|2}
Let $D = 48 \left\{  \log\left( \frac{\nsize \acsmaximum}{\epmaximum } \right) + \log \left( 1 + \frac{\vmaximum}{\flength} \right)  + 5 \right\}$. 
The maximum drift rate of the clock of any node is bounded by $\frac{1}{D}$ seconds/second.
\end{assumption}

We first prove the following lemma:

\begin{lemma}
\label{lem:drift|rate:small}
If \asmref{drift|rate|2} holds, then $i_0 \cdot \drmaximum \leq \frac{1}{16}$.
\end{lemma}
\begin{proof} We have:
\begin{align}
\nonumber
M
& = \frac{48 \max(2\acsmaximum, 3 \dmaximum_0)}{\srminimum} \ln\left( \frac{\nsize^2}{\epsilon} \right) + 1 \\
\nonumber
& \leq \left( \frac{48 \cdot (2\acsmaximum) \cdot (6 \nsize)}{\frac{1}{\acsmaximum}} \right) \cdot \left( \frac{\nsize^2}{\epsilon} \right) + 1 \\
\nonumber
& = \frac{576 \cdot \acsmaximum^2 \cdot \nsize^3}{\epmaximum} + 1
\end{align}
This implies that:
\begin{align}
\nonumber
\log M
& \leq \log 576 + 3 \log \left( \frac{\nsize \acsmaximum}{\epmaximum} \right) + 1 \\
& \leq 3 \log \left( \frac{\nsize \acsmaximum}{\epmaximum} \right)  + 11
\label{eq:logM}
\end{align}

Now, we have:
\begin{align}
\nonumber
i_0 \cdot \drmaximum 
& \leq \left\{ \log \left( 6 M + \frac{3\vmaximum}{\flength} \right) + 1 \right\} \cdot \drmaximum \\
\nonumber
& \leq \left\{ \log 6 + \log \left( M + \frac{\vmaximum}{\flength} + 1 \right) + 1 \right\} \cdot \drmaximum \\
\nonumber
& \phantom{\leq} \left\{ \text{for all } x, y \text{ with } x \geq 1 \text{ and } y \geq 1, \log (x + y) \leq \log (x y) \leq \log x + \log y \right\} \\
\nonumber
& \leq \left\{ \log(M) + \log\left( 1 + \frac{\vmaximum}{\flength} \right) + 4 \right\} \cdot \drmaximum \\
\nonumber
& \phantom{\leq} \left\{ \text{using \eqref{eq:logM}} \right\} \\
\nonumber
& \leq \left\{ 3 \log \left( \frac{\nsize \acsmaximum}{\epmaximum} \right)  + 11 + \log\left( 1 + \frac{\vmaximum}{\flength} \right) + 4 \right\} \cdot \drmaximum \\
\nonumber
& \leq 3 \cdot \left\{ \log \left( \frac{\nsize \acsmaximum}{\epmaximum} \right)  + \log\left( 1 + \frac{\vmaximum}{\flength} \right) + 5 \right\} \cdot \drmaximum \\
\nonumber
& \phantom{\leq} \left\{ \text{using definition of } D \right\} \\
\nonumber
& \leq 3 \cdot \frac{D}{48} \cdot \drmaximum \\
\nonumber
& \phantom{\leq} \left\{ \text{using \asmref{drift|rate|2}} \right\} \\
\nonumber
& \leq  3 \cdot \frac{D}{48} \cdot \frac{1}{D} = \frac{1}{16}
\end{align}

This establishes the lemma.
\end{proof}

Now, substituting $i = i_0$ and $j = j_0$ in \eqref{eq:phase|min} and \eqref{eq:phase|max} and subtracting \eqref{eq:phase|min} from \eqref{eq:phase|max}, we obtain:
\begin{align}
\nonumber
& \finishtimeOf{i_0}{j_0} - \starttimeOf{i_0}{j_0} \\
\nonumber
& \geq 
(\framesuntil{i_0}{j_0} + 2^{i_0}) \left( \frac{\flength}{1 + \drmaximum} \right) - \framesuntil{i_0}{j_0} \left( \frac{\flength}{1 - \drmaximum} \right) - \vmaximum \\
\nonumber
& \phantom{\geq} \left\{ \text{rearranging terms} \right\} \\
\nonumber
& = 2^{i_0} \left( \frac{\flength}{1 + \drmaximum} \right)  - \framesuntil{i_0}{j_0} \left( \frac{2 \drmaximum \flength}{1 - \drmaximum^2} \right) - \vmaximum \\
\nonumber
& \phantom{\geq} \left\{ \text{using \eqref{eq:frames|until}} \right\} \\
\nonumber
& = 2^{i_0} \left( \frac{\flength}{1 + \drmaximum} \right)  - (i_0 + j_0 - 2) \cdot 2^{i_0} \left( \frac{2 \drmaximum \flength}{1 - \drmaximum^2} \right) - \vmaximum \\
\nonumber
& \geq 2^{i_0} \left( \frac{\flength}{1 + \drmaximum} \right) \left\{  1  - (i_0 + j_0 - 2) \left( \frac{2 \drmaximum}{1 - \drmaximum} \right) \right\} - \vmaximum \\
\nonumber
& \phantom{\geq} \left\{ i_0 + j_0 - 2 \leq 2 i_0 \right\} \\
\nonumber
& \geq  2^{i_0} \left( \frac{\flength}{1 + \drmaximum} \right) \left(  1 - \frac{(2 i_0) (2\drmaximum)}{1 - \drmaximum} \right) - \vmaximum \\
\nonumber
& \phantom{\geq} \left\{ \drmaximum \leq \frac{1}{2} \right\} \\
\nonumber
& \geq  2^{i_0} \left( \frac{2\flength}{3} \right)  ( 1 - 8 \cdot i_0 \cdot \drmaximum) - \vmaximum  \\
\nonumber
& \phantom{\geq} \left\{ \text{using \lemref{drift|rate:small}} \right\} \\
\nonumber
& \geq 2^{i_0} \left( \frac{\flength}{3} \right)  - \vmaximum \\
\nonumber
& \geq \left( 6M + \frac{3 \vmaximum}{\flength} \right)  \left( \frac{\flength}{3} \right)  - \vmaximum \\
\nonumber 
& = 2 M \flength + \vmaximum - \vmaximum  = 2 M \flength
\end{align}

We now prove the main result.
Since the maximum length of a frame in real-time is $\frac{\flength}{1 - \drmaximum}$ and $\drmaximum \leq \frac{1}{2}$, it follows that the time-period from $\starttimeOf{i_0}{j_0}$ to $\starttimeOf{i_0}{j_0}$ 
contains at least $\frac{2 M \flength}{\nicefrac{\flength}{1  - \drmaximum}} = 2 M (1 - \drmaximum) \geq M$ full frames. 

The above analysis implies that, once all nodes have completed their epoch $i_0$, the neighbor discovery completes with probability at least $1 - \epmaximum$. Thus we have: 

\begin{theorem}
Let $M =  \frac{48 \max(2\acsmaximum, 3 \dmaximum_0)}{\srminimum} \ln\left( \frac{\nsize^2}{\epsilon} \right) + 1$. Then \algoref{asynchronous|none} ensures that each node discovers all its neighbors  on all channels  within $O\left( \left( M + \frac{\vmaximum}{\flength} \right) \log\left( M + \frac{\vmaximum}{\flength} \right) \left( \frac{\flength}{1 - \drmaximum} \right) \right)$ time-units of $\starttime$ with probability at least $1 - \epmaximum$.
\end{theorem}

\begin{remark}[Impact of \asmref{drift|rate|2}]
Even for very large values of system parameters, namely $\nsize = 10^6$, $\acsmaximum = 10^6$, $\epmaximum = 10^{-9}$ and $\frac{\vmaximum}{\flength} = 10^6$, \asmref{drift|rate|2} only requires $\drmaximum \leq 2.2 \times 10^-3$~seconds/seconds. This is not very restrictive because, as we mentioned earlier, an ordinary quartz clock has a maximum drift rate of $10^{-6}$ seconds/second.
\end{remark}

%% file: frame-local.pstex_t
\begin{picture}(0,0)%
\includegraphics{frame-local.pstex}%
\end{picture}%
\setlength{\unitlength}{3947sp}%
\begingroup\makeatletter\ifx\SetFigFont\undefined%
\gdef\SetFigFont#1#2#3#4#5{%
  \reset@font\fontsize{#1}{#2pt}%
  \fontfamily{#3}\fontseries{#4}\fontshape{#5}%
  \selectfont}%
\fi\endgroup%
\begin{picture}(15831,2450)(2379,-2789)
\put(4576,-2161){\makebox(0,0)[b]{\smash{{\SetFigFont{20}{24.0}{\familydefault}{\mddefault}{\updefault}{\color[rgb]{0,0,0}Frame of length $\flength$}%
}}}}
\put(9151,-2236){\makebox(0,0)[b]{\smash{{\SetFigFont{20}{24.0}{\familydefault}{\mddefault}{\updefault}{\color[rgb]{0,0,0}Slot of}%
}}}}
\put(9151,-2686){\makebox(0,0)[b]{\smash{{\SetFigFont{20}{24.0}{\familydefault}{\mddefault}{\updefault}{\color[rgb]{0,0,0}length $\tfrac{\flength}{3}$}%
}}}}
\end{picture}%

%% file: frame-global.pstex_t
\begin{picture}(0,0)%
\includegraphics{frame-global.pstex}%
\end{picture}%
\setlength{\unitlength}{3947sp}%
\begingroup\makeatletter\ifx\SetFigFont\undefined%
\gdef\SetFigFont#1#2#3#4#5{%
  \reset@font\fontsize{#1}{#2pt}%
  \fontfamily{#3}\fontseries{#4}\fontshape{#5}%
  \selectfont}%
\fi\endgroup%
\begin{picture}(19274,5014)(451,-4883)
\put(4366,-1921){\makebox(0,0)[b]{\smash{{\SetFigFont{20}{24.0}{\familydefault}{\mddefault}{\updefault}{\color[rgb]{0,0,0}$g_1$}%
}}}}
\put(8521,-1921){\makebox(0,0)[b]{\smash{{\SetFigFont{20}{24.0}{\familydefault}{\mddefault}{\updefault}{\color[rgb]{0,0,0}$g_2$}%
}}}}
\put(13306,-1921){\makebox(0,0)[b]{\smash{{\SetFigFont{20}{24.0}{\familydefault}{\mddefault}{\updefault}{\color[rgb]{0,0,0}$g_3$}%
}}}}
\put(7681,-3721){\makebox(0,0)[b]{\smash{{\SetFigFont{20}{24.0}{\familydefault}{\mddefault}{\updefault}{\color[rgb]{0,0,0}$h_1$}%
}}}}
\put(11731,-3721){\makebox(0,0)[b]{\smash{{\SetFigFont{20}{24.0}{\familydefault}{\mddefault}{\updefault}{\color[rgb]{0,0,0}$h_2$}%
}}}}
\put(15616,-3721){\makebox(0,0)[b]{\smash{{\SetFigFont{20}{24.0}{\familydefault}{\mddefault}{\updefault}{\color[rgb]{0,0,0}$h_3$}%
}}}}
\put(451,-841){\makebox(0,0)[lb]{\smash{{\SetFigFont{20}{24.0}{\familydefault}{\mddefault}{\updefault}{\color[rgb]{0,0,0}Node $v$}%
}}}}
\put(451,-2656){\makebox(0,0)[lb]{\smash{{\SetFigFont{20}{24.0}{\familydefault}{\mddefault}{\updefault}{\color[rgb]{0,0,0}Node $u$}%
}}}}
\put(451,-4441){\makebox(0,0)[lb]{\smash{{\SetFigFont{20}{24.0}{\familydefault}{\mddefault}{\updefault}{\color[rgb]{0,0,0}Node $w$}%
}}}}
\put(4636,-121){\makebox(0,0)[b]{\smash{{\SetFigFont{20}{24.0}{\familydefault}{\mddefault}{\updefault}{\color[rgb]{0,0,0}$f_1$}%
}}}}
\put(9091,-121){\makebox(0,0)[b]{\smash{{\SetFigFont{20}{24.0}{\familydefault}{\mddefault}{\updefault}{\color[rgb]{0,0,0}$f_2$}%
}}}}
\put(13696,-121){\makebox(0,0)[b]{\smash{{\SetFigFont{20}{24.0}{\familydefault}{\mddefault}{\updefault}{\color[rgb]{0,0,0}$f_3$}%
}}}}
\end{picture}%

%% file: slots.pstex_t
\begin{picture}(0,0)%
\includegraphics{slots.pstex}%
\end{picture}%
\setlength{\unitlength}{3947sp}%
\begingroup\makeatletter\ifx\SetFigFont\undefined%
\gdef\SetFigFont#1#2#3#4#5{%
  \reset@font\fontsize{#1}{#2pt}%
  \fontfamily{#3}\fontseries{#4}\fontshape{#5}%
  \selectfont}%
\fi\endgroup%
\begin{picture}(19409,5891)(-3839,-9174)
\put(10201,-9038){\makebox(0,0)[b]{\smash{{\SetFigFont{20}{24.0}{\familydefault}{\mddefault}{\updefault}{\color[rgb]{0,0,0}$b_5$}%
}}}}
\put(12016,-9038){\makebox(0,0)[b]{\smash{{\SetFigFont{20}{24.0}{\familydefault}{\mddefault}{\updefault}{\color[rgb]{0,0,0}$b_6$}%
}}}}
\put(6436,-4321){\makebox(0,0)[b]{\smash{{\SetFigFont{20}{24.0}{\familydefault}{\mddefault}{\updefault}{\color[rgb]{0,0,0}$f_1$}%
}}}}
\put(10891,-4321){\makebox(0,0)[b]{\smash{{\SetFigFont{20}{24.0}{\familydefault}{\mddefault}{\updefault}{\color[rgb]{0,0,0}$f_2$}%
}}}}
\put(2401,-3586){\makebox(0,0)[b]{\smash{{\SetFigFont{20}{24.0}{\familydefault}{\mddefault}{\updefault}{\color[rgb]{0,0,0}$T$}%
}}}}
\put(7111,-7051){\makebox(0,0)[b]{\smash{{\SetFigFont{20}{24.0}{\familydefault}{\mddefault}{\updefault}{\color[rgb]{0,0,0}$g_1$}%
}}}}
\put(11401,-7051){\makebox(0,0)[b]{\smash{{\SetFigFont{20}{24.0}{\familydefault}{\mddefault}{\updefault}{\color[rgb]{0,0,0}$g_2$}%
}}}}
\put(-3824,-5101){\makebox(0,0)[lb]{\smash{{\SetFigFont{20}{24.0}{\familydefault}{\mddefault}{\updefault}{\color[rgb]{0,0,0}Node $v$}%
}}}}
\put(-3824,-7786){\makebox(0,0)[lb]{\smash{{\SetFigFont{20}{24.0}{\familydefault}{\mddefault}{\updefault}{\color[rgb]{0,0,0}Node $u$}%
}}}}
\put(4201,-6331){\makebox(0,0)[b]{\smash{{\SetFigFont{20}{24.0}{\familydefault}{\mddefault}{\updefault}{\color[rgb]{0,0,0}$a_1$}%
}}}}
\put(5716,-6331){\makebox(0,0)[b]{\smash{{\SetFigFont{20}{24.0}{\familydefault}{\mddefault}{\updefault}{\color[rgb]{0,0,0}$a_2$}%
}}}}
\put(7201,-6331){\makebox(0,0)[b]{\smash{{\SetFigFont{20}{24.0}{\familydefault}{\mddefault}{\updefault}{\color[rgb]{0,0,0}$a_3$}%
}}}}
\put(8701,-6331){\makebox(0,0)[b]{\smash{{\SetFigFont{20}{24.0}{\familydefault}{\mddefault}{\updefault}{\color[rgb]{0,0,0}$a_4$}%
}}}}
\put(10216,-6331){\makebox(0,0)[b]{\smash{{\SetFigFont{20}{24.0}{\familydefault}{\mddefault}{\updefault}{\color[rgb]{0,0,0}$a_5$}%
}}}}
\put(11701,-6331){\makebox(0,0)[b]{\smash{{\SetFigFont{20}{24.0}{\familydefault}{\mddefault}{\updefault}{\color[rgb]{0,0,0}$a_6$}%
}}}}
\put(4651,-9038){\makebox(0,0)[b]{\smash{{\SetFigFont{20}{24.0}{\familydefault}{\mddefault}{\updefault}{\color[rgb]{0,0,0}$b_1$}%
}}}}
\put(6151,-9038){\makebox(0,0)[b]{\smash{{\SetFigFont{20}{24.0}{\familydefault}{\mddefault}{\updefault}{\color[rgb]{0,0,0}$b_2$}%
}}}}
\put(7801,-9038){\makebox(0,0)[b]{\smash{{\SetFigFont{20}{24.0}{\familydefault}{\mddefault}{\updefault}{\color[rgb]{0,0,0}$b_3$}%
}}}}
\put(9076,-9038){\makebox(0,0)[b]{\smash{{\SetFigFont{20}{24.0}{\familydefault}{\mddefault}{\updefault}{\color[rgb]{0,0,0}$b_4$}%
}}}}
\end{picture}%

%% file: cases.pstex_t
\begin{picture}(0,0)%
\includegraphics{cases.pstex}%
\end{picture}%
\setlength{\unitlength}{3947sp}%
\begingroup\makeatletter\ifx\SetFigFont\undefined%
\gdef\SetFigFont#1#2#3#4#5{%
  \reset@font\fontsize{#1}{#2pt}%
  \fontfamily{#3}\fontseries{#4}\fontshape{#5}%
  \selectfont}%
\fi\endgroup%
\begin{picture}(18493,10786)(-2810,-9610)
\put(10980,-9478){\makebox(0,0)[b]{\smash{{\SetFigFont{20}{24.0}{\familydefault}{\mddefault}{\updefault}{\color[rgb]{0,0,0}(d)}%
}}}}
\put(-2744,-1029){\makebox(0,0)[b]{\smash{{\SetFigFont{20}{24.0}{\familydefault}{\mddefault}{\updefault}{\color[rgb]{0,0,0}$a_1$}%
}}}}
\put(2956,-1029){\makebox(0,0)[b]{\smash{{\SetFigFont{20}{24.0}{\familydefault}{\mddefault}{\updefault}{\color[rgb]{0,0,0}$a_5$}%
}}}}
\put(-1544,-3166){\makebox(0,0)[b]{\smash{{\SetFigFont{20}{24.0}{\familydefault}{\mddefault}{\updefault}{\color[rgb]{0,0,0}$b_1$}%
}}}}
\put(976,-3871){\makebox(0,0)[b]{\smash{{\SetFigFont{20}{24.0}{\familydefault}{\mddefault}{\updefault}{\color[rgb]{0,0,0}(a)}%
}}}}
\put(-464,-6645){\makebox(0,0)[b]{\smash{{\SetFigFont{20}{24.0}{\familydefault}{\mddefault}{\updefault}{\color[rgb]{0,0,0}$a_1$}%
}}}}
\put(1126,-6645){\makebox(0,0)[b]{\smash{{\SetFigFont{20}{24.0}{\familydefault}{\mddefault}{\updefault}{\color[rgb]{0,0,0}$a_2$}%
}}}}
\put(-1259,-8805){\makebox(0,0)[b]{\smash{{\SetFigFont{20}{24.0}{\familydefault}{\mddefault}{\updefault}{\color[rgb]{0,0,0}$b_1$}%
}}}}
\put(3248,-8805){\makebox(0,0)[b]{\smash{{\SetFigFont{20}{24.0}{\familydefault}{\mddefault}{\updefault}{\color[rgb]{0,0,0}$b_4$}%
}}}}
\put(975,-9487){\makebox(0,0)[b]{\smash{{\SetFigFont{20}{24.0}{\familydefault}{\mddefault}{\updefault}{\color[rgb]{0,0,0}(b)}%
}}}}
\put(6811,-2828){\makebox(0,0)[b]{\smash{{\SetFigFont{20}{24.0}{\familydefault}{\mddefault}{\updefault}{\color[rgb]{0,0,0}$b_1$}%
}}}}
\put(11318,-2828){\makebox(0,0)[b]{\smash{{\SetFigFont{20}{24.0}{\familydefault}{\mddefault}{\updefault}{\color[rgb]{0,0,0}$b_4$}%
}}}}
\put(10396,-638){\makebox(0,0)[b]{\smash{{\SetFigFont{20}{24.0}{\familydefault}{\mddefault}{\updefault}{\color[rgb]{0,0,0}$a_1$}%
}}}}
\put(11986,-638){\makebox(0,0)[b]{\smash{{\SetFigFont{20}{24.0}{\familydefault}{\mddefault}{\updefault}{\color[rgb]{0,0,0}$a_2$}%
}}}}
\put(10981,-3862){\makebox(0,0)[b]{\smash{{\SetFigFont{20}{24.0}{\familydefault}{\mddefault}{\updefault}{\color[rgb]{0,0,0}(c)}%
}}}}
\put(11791,-6272){\makebox(0,0)[b]{\smash{{\SetFigFont{20}{24.0}{\familydefault}{\mddefault}{\updefault}{\color[rgb]{0,0,0}$a_1$}%
}}}}
\put(6811,-8410){\makebox(0,0)[b]{\smash{{\SetFigFont{20}{24.0}{\familydefault}{\mddefault}{\updefault}{\color[rgb]{0,0,0}$b_1$}%
}}}}
\put(11318,-8410){\makebox(0,0)[b]{\smash{{\SetFigFont{20}{24.0}{\familydefault}{\mddefault}{\updefault}{\color[rgb]{0,0,0}$b_4$}%
}}}}
\put(12968,-8410){\makebox(0,0)[b]{\smash{{\SetFigFont{20}{24.0}{\familydefault}{\mddefault}{\updefault}{\color[rgb]{0,0,0}$b_5$}%
}}}}
\end{picture}%

%% file: summary.tex
\section{Summary}
\label{sec:summary}

In this paper, we have presented several randomized algorithms for neighbor discovery in an \model{} network both for a synchronous and an asynchronous system, and also analyzed their running time.
\Tabref{summary} summarizes the results of the previous two sections.
We now discuss two extensions to our 
algorithms to enhance their applicability.

\begin{table}[t]
\begin{center}
\caption{A summary of all randomized neighbor discovery algorithms in this work.}
\label{tab:summary}
\fontsize{8}{9}
\selectfont
\begin{tabular}{|c|c|>{\centering\arraybackslash}m{1.75in}|>{\centering\arraybackslash}m{2.0in}|c|} \hline
\multicolumn{3}{|c|}{\textbf{Assumptions}} & \multirow{2}{*}{\textbf{Time Complexity$^\ast$}} & \multirow{2}{*}{\textbf{Section}} \\ \cline{1-3}
\textbf{Synchrony} & \textbf{Start Times} & \textbf{System Parameters} &  & \\ \hline \hline
Synchronous & Same & know an estimate for an upper bound on $\dmaximum$, given by $\destimate$ &  $O\left( \frac{\max( \acsmaximum, \dmaximum )}{\srminimum} \log(\destimate)\log\left( \frac{\nsize}{\epmaximum}\right)\right)$ & 
\ref{sec:identical|degree} \\ \hline
Synchronous & Same &  - & $O( M \log M )$ where 
$M = \frac{16 \max(\acsmaximum, \dmaximum)}{\srminimum}  \ln\left( \frac{\nsize^2}{\epmaximum} \right)$ & \ref{sec:identical|none}  \\ \hline
Synchronous & Different & know an estimate for an upper bound on $\dmaximum$, given by $\destimate$ &  $O\left( \frac{\max( 2\acsmaximum, \destimate )}{\srminimum} \log\left( \frac{\nsize}{\epmaximum}\right)\right)$ &
\ref{sec:different|degree} \\ \hline
Synchronous & Different & - & $O((M + \vmaximum)\log( M + \vmaximum ))$ where $M =  \frac{16 \max(\acsmaximum, \dmaximum_0)}{\srminimum}  \ln\left( 
\frac{\nsize^2}{\epmaximum} \right)$ &  \ref{sec:different|none} \\ \hline
Asynchronous & Different & know an estimate for an upper bound on $\dmaximum$, given by $\destimate$,  and $\drmaximum \leq \frac{1}{7}$ & $O\left(\frac{\flength \max(2\acsmaximum, 3 \destimate) }{\srminimum} \log\left( \frac{\nsize}{\epsilon} \right) \right)$ & \ref{sec:asynchronous|degree} \\ \hline
Asynchronous & Different & $\drmaximum \leq \frac{1}{48 \left\{  \log\left( \frac{\nsize \acsmaximum}{\epmaximum } \right) + \log \left( 1 + \frac{\vmaximum}{\flength} \right)  + 5 \right\}}$ & 
$O\left( \left( M \flength + \vmaximum \right) \log\left( M + \frac{\vmaximum}{\flength} \right) \right)$ where $M = \frac{48 \max(2\acsmaximum, 3 \dmaximum_0)}{\srminimum} \ln\left( \frac{\nsize^2}{\epsilon} \right) + 1$ & 
\ref{sec:asynchronous|none} \\ \hline
\multicolumn{5}{l}{$^\ast$: In the case of variable start times, the time is measured after the last node has started neighbor discovery} \\
\multicolumn{5}{l}{$\nsize$: total number of CR nodes in the network} \\
\multicolumn{3}{l}{$\dmaximum$: maximum degree of a node on a channel} &
\multicolumn{2}{l}{$\dmaximum_0$: smallest power of two greater than or equal to $\dmaximum$} \\
\multicolumn{3}{l}{$\acsmaximum$: size of the largest available channel set} &
\multicolumn{2}{l}{$\srminimum$: minimum span ratio of a link} \\
\multicolumn{5}{l}{$\epmaximum$: maximum error probability (user-specified)} \\
\multicolumn{5}{l}{$\vmaximum$: maximum difference in the start times of two nodes} \\
\multicolumn{3}{l}{$\flength$: length of a frame} &
\multicolumn{2}{l}{$\drmaximum$: maximum drift rate of a local clock} \\
\end{tabular}
\end{center}
\end{table}

%% file: extensions.tex
\section{Improving Robustness}
\label{sec:extensions}

\subsection{Handling Diverse Propagation Characteristics}
 
All algorithms described so far have assumed ``one-for-all'' property that if a link from node $v$ to node $u$ exists on some channel then it also exists on all channels that are present in available channel sets of both $u$ and $v$. That may not hold if different channels have very different propagation characteristics.  In  case the ``one-for-all'' property does not hold, we can partition the channels into \emph{bands} such that the subset of  channels within the same band have very similar propagation characteristics. As a result, ``one-for-all'' holds for each band individually. We say that a link can operate in a band if it can operate on some channel in the band. Each link can then be replaced with multiple links, one for each band it can operate in. Link for each band then has to be discovered separately. This modification increases the running time of our algorithms since more links have to be discovered now. The minimum span-ratio in the network now ranges from $\tfrac{1}{S}$ to $\tfrac{\bminimum}{\acsmaximum}$ instead of $\tfrac{1}{\acsmaximum}$ to $1$, where $\bminimum$ is the size of the smallest band (in terms of number of channels). Further the  term ``$\ln\left( \tfrac{\nsize^2}{\epmaximum} \right)$'' in time complexity expressions is replaced with the term ``$\ln\left( \tfrac{N^2 \bsize}{\epmaximum} \right)$'', where $\bsize$ denote the maximum number of bands in which any link can operate.

\subsection{Handling Lossy Channels}

All algorithms described so far have assumed that channels are reliable in the sense that a node fails to receive a message from its neighbor only due to collision. However, in the real world, a node may fail to receive a message due
to other reasons as well such as Gaussian noise or transient interference. Let $\fpmaximum$ with $0 \leq \fpmaximum < 1$ denote the probability that a node fails to receive a message
from its neighbor even in the absence of collision. Nodes may not know $\fpmaximum$. With lossy channels, the probability that a link is covered by a time-slot (in \algoref{identical|degree}, \algoref{identical|nodegree}, \algoref{variable|degree} and \algoref{variable|none}) or an \legal{} frame-pair (in \algoref{asynchronous|degree} and \algoref{asynchronous|none})  is now multiplied by a factor of $1 - \fpmaximum$. For example, the probability expression in \lemref{legal|atleast} now becomes 
 $\tfrac{\srminimum (1 - \fpmaximum)}{8\max(2\acsmaximum, 3\destimate)}$. It can be verified that this has the effect of increasing the running time of an algorithm by a factor of $\tfrac{1}{1 - \fpmaximum}$.

\subsection{Handling Jamming Attacks}

Most of the work on jamming 
attack~\cite{WanWu+:2011:JSAC,LiCad:2011:CISS,CheSon+:2013:IN,DabBet+:2014:ICASSP} 
is concerned with analyzing or mitigating the impact of jamming attack on 
communication throughput in the system, and, as such, they assume that neighbor discovery has already been done. 
Work in~\cite{AstZor+:2010:CCW} presents a suite of neighbor discovery 
algorithms in the presence of jamming attack. However, unlike our work, the work 
in~\cite{AstZor+:2010:CCW} assumes a single-hop network and does not provide 
any theoretical guarantees (algorithms are only evaluated using simulation experiments).

We show that, under certain conditions, our neighbor discovery algorithms, with minor modifications, are tolerant to jamming attacks in the sense that a malicious jammer who seeks to prevent neighbor discovery can only delay time to the completion of neighbor discovery by at most a constant factor.

In formulating the jamming model, we seek a balance between a jammer who is too powerful (making the neighbor discovery problem impossible to solve) and one who is too weak 
(making the attack easy to foil). For instance, if the jammer can simultaneously jam all channels at all times, then neighbor discovery would be unable to complete. Likewise, if for some pair of neighboring nodes, $\card{ \lspan{u}{v}} = 1$, then the jammer can ensure that discovery does not complete by simply jamming this channel. 

We assume an ``intelligent'' jammer who can \emph{scan} the spectrum and learn which channels are currently under use by transmitting nodes running the neighbor discovery algorithm. 
The jammer can then decide to jam any single channel, thereby disrupting possibly all transmissions on that channel. 
When scanning, the jammer only learns which channels are being used; we assume that it cannot determine which channels are used by a majority of nodes by analyzing the relative strength of the transmissions.

In a heterogeneous system, it may be possible for a jammer to increase its effectiveness by keeping statistics on which channels are seen in scans more often. In some cases, this would correspond to channels used by more nodes. It is not clear how best to analyze this case, so we explicitly avoid it by directing our analysis towards homogeneous systems in which $\card{\lspan{u}{v}} = \acsmaximum$ for all neighboring nodes $u$ and $v$. 

The jammer can change the channel which it jams during the course of the algorithm. Clearly, if the jammer is able to switch channels instantaneously, it would be possible for the jammer to block every channel in every slot making it as powerful as the one who can simultaneously jam all channels at all times. To avoid this, we assume a lower bound on the time between switching channels for jamming. In the asynchronous case, this lower bound also allows to assume that the outcomes of two consecutive scans are independent of each other. 

Note that the jammer need not respect slot boundaries even when the system is synchronous. Therefore the jammer can effectively jam \emph{two channels} in a given slot by selecting a new channel to jam near the middle of a slot. 

We also assume that the jammer has complete knowledge of our algorithm, but is ignorant of the exact random choices which are made. Clearly, the jammer's best option is to make a random selection of channels to jam. However, since the system is homogeneous, the best that the jammer can do is make a selection at random, in an attempt to foil the random choices made by our algorithm. 

Based on the discussion above, we now formulate our jamming model.
The jammer operates in rounds. At the beginning of each round, the jammer scans the spectrum and selects a channel to jam during that round from among the channels it observed to be in use in the scan. Specifically, let $\oldchannel$ denote the channel jammed by the jammer in the previous round. 
Also, let $\busychannels$ denote the set of channels the jammer found to be in use in the scan it performed at the beginning of the current round. 
The jammer then makes a uniform random selection of a channel from the set $\busychannels \setminus \oldchannel$, if non-empty, and jams this channel during the current round. We assume that round length is at least at large as slot length, but slot boundaries and round boundaries need not be aligned even if the system is synchronous. (A jamming round can start in the middle of a slot.) Note that it is not necessary for the jammer to jam $\oldchannel$ again since it would have already disrupted all transmissions on channel $\oldchannel$ in the current slot.

Note that, in the above-described jamming model, the neighbor discovery problem can only be solved if there are at least three nodes and at least three channels.
We now consider synchronous and asynchronous cases separately.

\subsubsection{Synchronous Algorithms}

Recall the events $\A{\tau}{c}$, $\B{\tau}{c}$ and $\C{\tau}{c}$ defined 
earlier when the analyzing the probability that node $u$ discovers node $v$ on channel $c$ during time-slot $\tau$. 
We  define a new event $\D{\tau}{c}$ to represent the condition that the 
transmission from node $v$ to node $u$ on channel $c$ is not jammed during 
the time-slot $\tau$. For convenience, let $\ABC{\tau}{c}$ denote the 
conjunction of the three events $\A{\tau}{c}$, $\B{\tau}{c}$ and $\C{\tau}{c}$. 
We derive a lower bound on the conditional probability that event 
$\D{\tau}{c}$ occurs given that event $\ABC{\tau}{c}$ occurs. 
Note that the \emph{slow-down factor} in the running time of a neighbor discovery algorithm that 
is undergoing a jamming attack is given by the \emph{inverse of this conditional probability}.

Assume that, at the beginning of the time-slot $\tau$, the jammer was jamming channel $\oldchannel$. Sometime during $\tau$, the jammer switches to a new channel to jam. We have:
\begin{align}
\Pr\{ \D{\tau}{c} \given \ABC{\tau}{c} \} \quad = \quad \left( \frac{\acsmaximum - 1}{\acsmaximum} \right) \cdot \Pr\{ \E{\tau}{c} \given \ABC{\tau}{c} \} \cdot \left( \frac{1}{2} \right)
\label{eq:jamming:notblocked}
\end{align}
where $\E{\tau}{c}$ denotes the event that some node in the system transmits on a channel other than $c$ and $\oldchannel$ during the time-slot $\tau$ with $c \neq \oldchannel$. In the above expression, the first term (given by $\frac{\acsmaximum-1}{\acsmaximum}$) captures the probability that $c \neq \oldchannel$ and the third term (given by $\frac{1}{2}$) the probability that the jammer chooses a channel other than $c$ to jam.

We next compute the conditional probability that the event $\E{\tau}{c}$ occurs given that the event $\ABC{\tau}{c}$ occurs. We partition the set of nodes in the system into three mutually-exclusive groups: $G_1$, $G_2$ and $G_3$. Group $G_1$ contains nodes $u$ and $v$, group $G_2$ contains nodes that are neighbors 
of node $u$ except for node $v$, and group $G_3$ contains the remaining nodes. 

Let $\acs$ denote the set of all available channels. (Note that, by our assumption, all nodes have identical available channel sets.) Also, let $\EBAR[w]{\tau}{c}$ denote the event that node $w$ \emph{does not} transmit on a channel in $\acs \setminus \{ c, \oldchannel \}$ during time-slot $\tau$. Since nodes act independently of each other, we have:
\begin{align}
\nonumber
\Pr\{ \E{\tau}{c} \given \ABC{\tau}{c}  \} 
& = 1 - \prod_{i \in \{ 1, 2, 3 \}} \prod_{w \in G_i} \Pr\{ \EBAR[w]{\tau}{c} \given \ABC{\tau}{c} \}  \\
\nonumber
& \phantom{=} \big\{ \forall w : w \in G_1 : \Pr\{ \EBAR[w]{\tau}{c} \given \ABC{\tau}{c} \} = 1 \big\} \\
\nonumber
& = 1 - \prod_{i \in \{ 2, 3 \}} \prod_{w \in G_i} \Pr\{ \EBAR[w]{\tau}{c} \given \ABC{\tau}{c} \}  \\
\nonumber
& \phantom{=} \left\{ \begin{array}{@{}l@{}} \forall w : w \in G_3 : \Pr\{ \EBAR[w]{\tau}{c} \given \ABC{\tau}{c} \} = \Pr\{ \EBAR[w]{\tau}{c} \}  \\[-0.75em]
\because \text{ nodes act independently} \end{array} \right\} \\
\label{eq:jamming:third|channel}
& = 1 -  \left( \displaystyle\prod_{w \in G_2} \Pr\{ \EBAR[w]{\tau}{c} \given \ABC{\tau}{c} \} \right) \cdot \left( \displaystyle\prod_{w \in G_3} \Pr\{ \EBAR[w]{\tau}{c} \} \right)  
\end{align}

Let $p$ denote the placeholder for transmission probability. Note that $p = \min\left( \frac{1}{2}, \frac{\acsmaximum}{\destimate} \right)$. For our analysis, we focus on time-slots for which $\destimate$ satisfies the following condition:
\begin{align}
\label{eq:assumption:degree}
\destimate \leq 2(\nsize -2)
\end{align}
We show that the above-condition trivially holds for three out of four synchronous algorithms, namely 
\algoref{identical|degree}, \algoref{identical|nodegree} and \algoref{variable|none}. For \algoref{variable|degree}, in which $\destimate$ is an ``external'' parameter, we state it as an explicit assumption. In the three synchronous algorithms (in which nodes estimate degree), for the time-slots used to show that neighbor discovery completes with arbitrarily high probability, note that $\destimate$ is the smallest power of two that is greater than or equal to either $\degreeon{u}{c}$ or $\dmaximum$. In either case, we have:
\begin{align*}
& \phantom{\implies} \frac{\destimate}{2} < \dmaximum \\
& \implies \frac{\destimate}{2} \leq \dmaximum - 1 \\
& \phantom{\implies} \big\{ \dmaximum \leq \nsize - 1 \big\} \\
& \implies \frac{\destimate}{2} \leq (\nsize - 1) - 1 = \nsize - 2
\end{align*}

We now derive the conditional probability $\Pr\{ \EBAR[w]{\tau}{c} \given \ABC{\tau}{c} \}$. There are two cases depending on whether $w \in G_2$ or $w \in G_3$.
First, assume that $w \in G_2$. We have:
\begin{align}
\nonumber
\Pr\{ \EBAR[w]{\tau}{c} \given \ABC{\tau}{c} \} 
& = \Pr\{  w \text{ does not transmit on } \acs \setminus \{ c, \oldchannel \}  \given  w \text{ does not transmit on } c \} \\
\nonumber
& \phantom{=} \left\{  \Pr\{ X \given Y \} = \frac{\Pr\{ X \land  Y \}}{\Pr\{ Y \}} \right\} \\
\nonumber
& = \frac{\Pr\{ (w \text{ does not transmit on } \acs \setminus \{ c, \oldchannel \}) \land  (w \text{ does not transmit on } c)\}}{\Pr\{ w \text{ does not transmit on } c \}} \\
\nonumber
& = \frac{\Pr\{ (w \text{ listens) or (} w \text{ transmits on channel } \oldchannel) \}}
{\Pr\{ (w \text{ listens) or (} w \text{ transmits on a channel other than } c) \}} \\
\nonumber
& = \frac{1 - p + p \cdot \frac{1}{\acsmaximum}}{1 - p + p \cdot \frac{\acsmaximum-1}{\acsmaximum} } \\
\nonumber
& = \frac{1 - p \frac{\acsmaximum-1}{\acsmaximum}}{1 - \frac{p}{\acsmaximum}} \\
\nonumber
& \phantom{=} \big\{ \text{multiplying numerator and denominator by } \frac{\acsmaximum}{p} \big\} \\
\nonumber
& = \frac{\frac{\acsmaximum}{p} - \acsmaximum + 1}{\frac{\acsmaximum}{p} - 1} \\
\label{eq:jamming:G2}
& = 1 - \frac{\acsmaximum - 2}{\frac{\acsmaximum}{p} - 1} 
\end{align}

Next, assume that $w \in G_3$. We have:
\begin{align}
\nonumber
\Pr\{ \EBAR[w]{\tau}{c} \given \ABC{\tau}{c} \} 
& = \Pr\{ w \text{ does not transmit on } \acs \setminus \{ c, \oldchannel \} \} \\
\nonumber
& = \Pr\{ (w \text{ listens) or (} w \text{ transmits on channel } c \text{ or } \oldchannel) \} \\
\nonumber
& = 1 - p + p \cdot \frac{2}{\acsmaximum} \\
\nonumber
& = 1 - p \cdot \frac{\acsmaximum-2}{\acsmaximum} \\
\label{eq:jamming:G3}
& = 1 - \frac{\acsmaximum - 2}{\frac{\acsmaximum}{p}} 
\end{align}

Using \eqref{eq:jamming:third|channel}, \eqref{eq:jamming:G2} and \eqref{eq:jamming:G3}, we have:
\begin{align}
\nonumber
\Pr\{ \E{\tau}{c} \given \ABC{\tau}{c} \} 
& = 1 -   
\displaystyle\prod_{w \in G_2} \left(1 - \frac{\acsmaximum-2}{\frac{\acsmaximum}{p}-1} \right)   \cdot 
\displaystyle\prod_{w \in G_3} \left( 1 - \frac{\acsmaximum-2}{\frac{\acsmaximum}{p}} \right) \\
\nonumber
& \geq 1 - \displaystyle\prod_{w \in G_2} \left(1 - \frac{\acsmaximum-2}{\frac{\acsmaximum}{p}} \right) \cdot 
\displaystyle\prod_{w \in G_3} \left( 1 - \frac{\acsmaximum-2}{\frac{\acsmaximum}{p}} \right) \\
\nonumber
& = 1 - \displaystyle\prod_{w \in G_2 \cup G_3} \left(1 - \frac{\acsmaximum-2}{\frac{\acsmaximum}{p}} \right) \\
\label{eq:jamming:general}
& = 1 - \left(1 - p \cdot \frac{\acsmaximum-2}{\acsmaximum} \right)^{\nsize - 2} 
\end{align}

We analyze the cases corresponding to when $p = \frac{1}{2}$ and when $p = \frac{\acsmaximum}{\destimate}$ separately. 

\begin{itemize}

\item \textbf{Case 1 ($p = \frac{1}{2}$):} Using \eqref{eq:jamming:notblocked} and \eqref{eq:jamming:general}, we have:
\begin{align*}
\Pr\{ \D{\tau}{c} \given \ABC{\tau}{c} \} 
& \geq \frac{1}{2} \cdot \left( \frac{\acsmaximum-1}{\acsmaximum} \right) \cdot \left\{ 1 - \left(1 - \frac{\acsmaximum-2}{2\acsmaximum} \right)^{\nsize - 2} \right\} \\
& \phantom{\geq} \big\{ \nsize \geq 3 \big\} \\
& \geq \frac{1}{2} \cdot \left( \frac{\acsmaximum-1}{\acsmaximum} \right) \cdot \left( \frac{\acsmaximum-2}{2\acsmaximum} \right) \\
& = \frac{1}{4} \cdot \left( 1 -  \frac{1}{\acsmaximum} \right) \cdot \left( 1 - \frac{2}{\acsmaximum} \right) \\
& \phantom{\geq} \big\{ \acsmaximum \geq 3 \big\} \\
& \geq \frac{1}{4} \cdot \frac{2}{3} \cdot \frac{1}{3} = \frac{1}{18}
\end{align*}

\item \textbf{Case 2 ($p = \frac{\acsmaximum}{\destimate}$):}  Using \eqref{eq:jamming:general}, we have:
\begin{align}
\nonumber
\Pr\{ \E{\tau}{c} \given \ABC{\tau}{c} \} 
\nonumber
& = 1 - \left(1 - \frac{\acsmaximum-2}{\destimate} \right)^{\nsize -2} \\
\nonumber
& \phantom{=} \left\{ \begin{array}{@{}l@{}} \text{substitute } x = \frac{\destimate}{\acsmaximum-2} \\
p = \frac{\acsmaximum}{\destimate} \implies \frac{\acsmaximum}{\destimate} \leq \frac{1}{2} 
\implies \frac{\destimate}{\acsmaximum} \geq 2 \implies x \geq 2 \end{array} \right\} \\
\nonumber
& \geq 1 - \left(1 - \frac{1}{x}\right)^{\nsize -2} \\
\nonumber
& \phantom{\geq} \left\{ \text{using \eqref{eq:assumption:degree}, } \nsize - 2 \geq \frac{\destimate}{2} \implies \nsize - 2 \geq \frac{x (\acsmaximum - 2)}{2} \right\} \\
\nonumber
& \geq 1 - \left\{ \left(1 - \frac{1}{x}\right)^x \right\}^{\frac{\acsmaximum-2}{2}} \\
\nonumber
& \phantom{\geq} \left\{ \left( 1 - \frac{1}{x} \right)^x \leq \frac{1}{\eulernum} \text{ when } x \geq 2 \right\} \\
\label{eq:jamming:case:2}
& \geq 1 - \left(\frac{1}{\eulernum}\right)^{\frac{\acsmaximum-2}{2}} 
\end{align}

Using \eqref{eq:jamming:notblocked} and \eqref{eq:jamming:case:2}, we have:
\begin{align*}
\Pr\{ \D{\tau}{c} \given \ABC{\tau}{c} \} 
& \geq \frac{1}{2} \cdot \left( \frac{\acsmaximum-1}{\acsmaximum} \right) \cdot \left\{ 1 - \left(\frac{1}{\eulernum}\right)^{\frac{\acsmaximum-2}{2}} \right\} \\
& \geq \frac{1}{2} \cdot \left( 1 -  \frac{1}{\acsmaximum} \right) \cdot \left( 1 - \frac{1}{\sqrt{\eulernum}} \right) \\
& \geq \frac{1}{2} \cdot \frac{2}{3} \cdot \frac{1}{3} = \frac{1}{9}
\end{align*}

\end{itemize}

Let $\sdf{\nsize}{\acsmaximum}$ denote the slow-down factor experienced by a neighbor discovery algorithm that is undergoing a jamming attack.
The above analysis implies that $\sdf{\nsize}{\acsmaximum} \leq 18$ for all $\nsize \geq 3$ and $\acsmaximum \geq 3$. For sufficiently large
 values of $\nsize$ and $\acsmaximum$, we show that the slow-down factor is much smaller. 
Again, we analyze the cases corresponding to when $p = \frac{1}{2}$ and when $p = \frac{\acsmaximum}{\destimate}$ separately. 

\begin{itemize}

\item \textbf{Case 1 ($p = \frac{1}{2}$):} Using \eqref{eq:jamming:notblocked} and \eqref{eq:jamming:general} we have:
\begin{align*}
\Pr\{ \D{\tau}{c} \given \ABC{\tau}{c} \} 
& \geq \frac{1}{2} \cdot \left( \frac{\acsmaximum-1}{\acsmaximum} \right) \cdot \left\{  1 - \left(1 -   \frac{\acsmaximum-2}{2\acsmaximum} \right)^{\nsize - 2} \right\} \\
& \phantom{\geq} \left\{ \text{for sufficiently large } \acsmaximum, \frac{\acsmaximum - 1}{\acsmaximum} \approx 1 \text{ and } \frac{\acsmaximum - 2}{2\acsmaximum} \approx \frac{1}{2} \right\} \\
& \approx \frac{1}{2} \cdot 1 \cdot \left\{ 1 - \left(\frac{1}{2}\right)^{\nsize - 2} \right\} \\
& \phantom{\approx}  \left\{ \text{for sufficiently large } \nsize, \left(\frac{1}{2}\right)^{\nsize - 2} \approx 0 \right\} \\
& \approx \frac{1}{2} \cdot (1 - 0) = \frac{1}{2}
\end{align*}

\item \textbf{Case 2 ($p = \frac{\acsmaximum}{\destimate}$):} Using \eqref{eq:jamming:notblocked} and \eqref{eq:jamming:case:2}, we have:
\begin{align*}
\Pr\{ \D{\tau}{c} \given \ABC{\tau}{c} \} 
& \geq \frac{1}{2} \cdot \left( \frac{\acsmaximum-1}{\acsmaximum} \right) \cdot 
\left\{ 1 - \left(\frac{1}{\eulernum}\right)^{\frac{\acsmaximum-2}{2}} \right\} \\
& \phantom{\geq} \left\{ \text{for sufficiently large } \acsmaximum, \frac{\acsmaximum - 1}{\acsmaximum} \approx 1 \text{ and } 
\left(\frac{1}{\eulernum}\right)^{\frac{\acsmaximum-2}{2}} \approx 0 \right\} \\
& \approx \frac{1}{2} \cdot 1 \cdot (1 - 0) = \frac{1}{2}
\end{align*}

\end{itemize}

The above analysis implies that:
\begin{align*}
\lim_{\substack{\nsize \to \infty \\ \acsmaximum \to \infty}} \sdf{\nsize}{\acsmaximum} \leq 2
\end{align*}

Finally, we have:

\begin{theorem}
The running time of the synchronous algorithms increases by a factor of 18 in the worst case. Furthermore, when $\nsize$ and $\acsmaximum$ are sufficient large, the running time only increases by a factor of at most 2 in the worst case.
\end{theorem}

Note that the slow-down factor of 18 represents the worst-case and only occurs for a ``small'' system which is easier to jam. As $\nsize$ and $\acsmaximum$, increases, the slow-down factor decreases rapidly as depicted in \figref{jamming}.

\begin{figure}[t]
\centering
\begin{tikzpicture}[scale=1.0, every node/.style={transform shape}]
\begin{groupplot}[group style = {group size = 2 by 1,horizontal sep=1.0in},
width=3.0in, 
height=2.75in,
domain=3:10,
every axis/.append style={font=\footnotesize},
xmin=3,
xmax=10,
ymin=1,
xlabel=$\acsmaximum$,
xlabel style = {font=\scriptsize,yshift=0.5ex}, 
ylabel style = {font=\scriptsize,xshift=0.5ex},
legend style={at={(0.95,0.875)},anchor=north east}]
\nextgroupplot[title={Synchronous Algorithms},ylabel=upper-bound on $\sdf{\nsize}{\acsmaximum}$,extra y ticks = {2, 18}]
\foreach \yindex in {3,4,8}
{
   \addplot+[mark repeat=4] {1/(0.5*((x-1)/x)*(1-(1-(min(0.5,x/(2*(\yindex-2))))*((x-2)/x))^(\yindex-2)))};
}
\addplot+[mark=none,thin,dashed,] {2};
\addplot+[mark=none,thin,dashed,] {18};
\legend{$\nsize =3$,$\nsize = 4$, $\nsize = 8$}
\nextgroupplot[title={Asynchronous Algorithms},ylabel=upper-bound on $\sdf{\nsize}{\acsmaximum}$,extra y ticks = {2, 21},]
\foreach \yindex in {3,4,8}
{
   \addplot+[mark repeat=4] {1/(0.5*((x-1)/x)*(1-(1-(min(0.5,x/(6*(\yindex-2))))*((x-2)/x))^(\yindex-2)))};
}
\addplot+[mark=none,thin,dashed,] {2};
\addplot+[mark=none,thin,dashed,] {21};
\legend{$\nsize =3$,$\nsize = 4$, $\nsize = 8$}
\end{groupplot}
\end{tikzpicture}
\caption{Upper bound on slow-down factor as a function of $\nsize$ and $\acsmaximum$.}
\label{fig:jamming}
\end{figure}
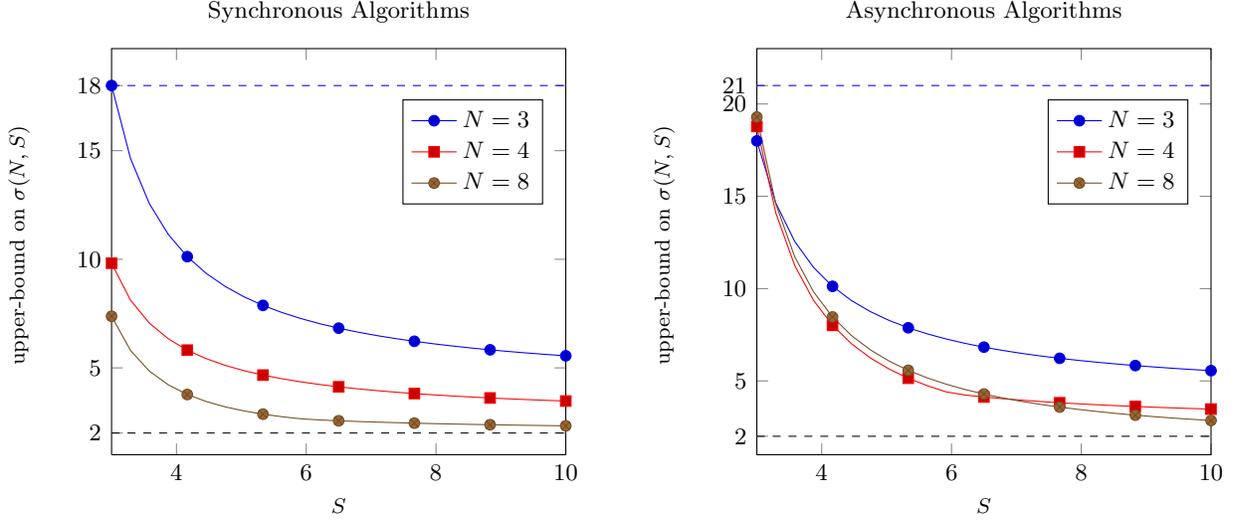

\subsubsection{Asynchronous Algorithms}

In the asynchronous case, we require a minor change to \algoref{asynchronous|degree} and 
\algoref{asynchronous|none}. Namely, when a node is in transmit mode, it makes a 
uniform selection of channel at the beginning of each slot, rather than at 
the beginning of each frame. We also assume that the length of a jamming round is at least the maximum length of a slot (measured in real-time) (which is given by $\frac{\flength}{3 (1 - \drmaximum)}$).

Consider the link from node $v$ to node $u$. Consider a legal pair 
of frames $\lpair{f}{g}$ with $\node{f} = v$ and $\node{g} = u$. Let 
$\slot{f}{g}$ denote the slot of $f$ that lies completely within $g$. (If 
there are multiple such slots, then choose one arbitrarily.) For analyzing 
the jamming attack, we narrow the period during which no transmission by other neighbors 
of $u$ interferes with $v$'s transmission during $\slot{f}{g}$. This, in 
conjunction with the assumption made on the minimum length of jamming round, 
ensures that the set of channels a jammer observes to be in use in two 
different scans do not have any correlation. Intuitively, this is because two consecutive 
scans by a jammer will involve two different slots of any node and, in transmit mode, a node 
now chooses the channel on which to transmit at the beginning of a slot (rather than 
a frame).

Consider events $\AF{f}{g}{c}$, $\BF{f}{g}{c}$ and $\CF{f}{g}{c}$. To analyze our algorithms under jamming attack, 
we replace the event $\AF{f}{g}{c}$ with a new event $\AFS{f}{g}{c}$ which 
denotes the condition that $u$ transmits on $c$ during $\slot{f}{g}$. 
Further, we replace the event $\CF{f}{g}{c}$ with a new event $\CFS{f}{g}{c}$ 
which denotes the condition that no neighbor of $u$ besides $v$
transmits on $c$ during any of its \emph{slots} that overlaps with 
$\slot{f}{g}$. Clearly, in the absence of jamming attack, $u$ will discover $v$ during 
frame $g$ if events $\AFS{f}{g}{c}$, $\BF{f}{g}{c}$ and $\CFS{f}{g}{c}$ occur.
Analogous to \lemref{overlap}, 
it can be easily shown that a single slot of a node can overlap with at most 
three slots of another node. Note that these three slots can come from at most 
two different frames. It can be verified that \eqref{eq:AA|1} holds for $\AFS{f}{g}{c}$ and \eqref{eq:CC|1} holds for $\CFS{f}{g}{c}$
(for the latter, see \secref{CC|slot} in the appendix).  

Analogous to events $\ABC{\tau}{c}$, $\D{\tau}{c}$ and $\E{\tau}{c}$ and $\EBAR[w]{\tau}{c}$ define in the synchronous case, 
we can define events $\ABCF{f}{g}{c}$, $\DF{f}{g}{c}$, $\EF{f}{g}{c}$ and $\EBARF{f}{g}{c}$, respectively.
Most of the analysis for the synchronous case is also applicable to the asynchronous case. Specifically, when $p = \frac{1}{2}$, the analysis and results for the synchronous case also apply to the asynchronous case. However, when $p = \frac{\acsmaximum}{3 \destimate}$, the results for the asynchronous case are slightly different from that for the synchronous case. 
Since the analysis is very similar to that for the synchronous case, we only present the main steps. 
Using an expression similar to \eqref{eq:jamming:general}, when $p = \frac{\acsmaximum}{3\destimate}$, we can show the following:
\begin{align*}
\Pr\{ \EF{f}{g}{c} \given \ABCF{f}{g}{c} \} 
& = 1 - \left(1 - \frac{\acsmaximum-2}{3\destimate} \right)^{\nsize -2} \\
& \geq 1 - \left(\frac{1}{\eulernum}\right)^{\frac{\acsmaximum-2}{6}} 
\end{align*}

For arbitrary $\nsize$ and $\acsmaximum$ with $\nsize, \acsmaximum \geq 3$, we have:
\begin{align*}
\Pr\{ \DF{f}{g}{c} \given \ABCF{f}{g}{c} \} 
& \geq \frac{1}{2} \cdot \left( \frac{\acsmaximum-1}{\acsmaximum} \right) \cdot \left\{ 1 - \left(\frac{1}{\eulernum}\right)^{\frac{\acsmaximum-2}{6}} \right\} \\
& \geq \frac{1}{2} \cdot \left( 1 -  \frac{1}{\acsmaximum} \right) \cdot \left( 1 - \frac{1}{\sqrt[6]{\eulernum}} \right) \\
& = \frac{1}{2} \cdot \frac{2}{3} \cdot \frac{1}{7} = \frac{1}{21}
\end{align*}

For sufficiently large $\nsize$ and $\acsmaximum$, we have:
\begin{align*}
\Pr\{ \DF{f}{g}{c} \given \ABCF{f}{g}{c} \} 
& \geq \frac{1}{2} \cdot \left( \frac{\acsmaximum-1}{\acsmaximum} \right) \cdot \left\{ 1 - \left(\frac{1}{\eulernum}\right)^{\frac{\acsmaximum-2}{6}} \right\} \\
& \gtrapprox \frac{1}{2} \cdot 1 \cdot (1 - 0) = \frac{1}{2}
\end{align*}

The above analysis implies that:
\begin{align*}
\forall \nsize, \forall \acsmaximum: \nsize \geq 3, \acsmaximum \geq 3 :  \sdf{\nsize}{\acsmaximum} \leq 21 
\qquad \text{ and } \qquad
\lim_{\substack{\nsize \to \infty \\ \acsmaximum \to \infty}} \sdf{\nsize}{\acsmaximum} \leq 2
\end{align*}

Finally, we have the following result:

\begin{theorem}
The running time of the asynchronous algorithms increases by a factor of 21 in the worst case. Furthermore, when $\nsize$ and $\acsmaximum$ are sufficient large, the running time only increases by a factor of at most 2 in the worst case.
\end{theorem}

Again, as in the synchronous case, the slow-down factor of 21 represents the worst-case and only occurs for a ``small'' system which is easier to jam. As $\nsize$ and $\acsmaximum$, increases, the slow-down factor decreases rapidly as depicted in \figref{jamming}.

\subsubsection{Discussion}

It is possible for a jammer to jam a band of channels and not just a single 
channel. In that case, we require the the available channel set to contain at 
least three \emph{orthogonal bands} of channels. The results derived above still hold if the 
analysis is conducted in terms of bands rather than channels and all bands are of the same size. 

It is also possible to analyze the running time of our algorithms when a 
jammer can jam up to $k$ channels in a jamming round for a fixed $k$, where $1
 \leq k < \frac{\acsmaximum}{2}$. The analysis is more involved and requires 
stronger assumptions on system size. The slow-down factor also increases to $O(k)$ and 
is no longer a constant.

%% file: relatedWork.tex


\section{Related Work}
\label{sec:related}

A large number of neighbor discovery algorithms have been proposed in the literature. To our knowledge, most of the algorithms suffer from one or more of the following limitations:

\begin{itemize}

\item \emph{Single-Channel Network:}
Several neighbor discovery algorithms have been proposed for a (traditional) single channel wireless network
(\emph{e.g.}, \cite{McGBor:2001:MOBIHOC,AloKra+:2003:IPDPS,ZheHou+:2003:MOBIHOC,VasKur+:2005:INFOCOM,BorEph+:2007:AN,AnHek:2007:ISMWC,DyoMas:2008:DCOSS,DutCul:2008:SenSys,HamChe+:2008:DMTCS,VasTow+:2009:MOBICOM,YanShi+:2009:PERCOM,VasAdl+:2013:TN,ZhaLuo+:2013:PE,RusVas+:2014:TPDS,WanMao+:2015:TPDS}).  
Some of these neighbor discovery algorithms, such as those proposed in \cite{VasTow+:2009:MOBICOM,VasAdl+:2013:TN}, 
can be extended to work for a multi-channel network (including a heterogeneous network). The main idea is as follows.
Let the collective set of all channels over which radio nodes in the network are capable of operating be referred to as the \emph{universal channel set}. 
The main idea is to execute a separate instance of single-channel neighbor discovery algorithm  on all channels in the universal channel set \emph{concurrently} by interleaving their executions.
A node only participates in instances that are associated with channels in its available channel set.
However, this simple approach has several disadvantages.
First, it requires that all nodes have to agree on the composition of the universal channel set.
Second, the time complexity of the algorithm for multi-channel network (obtained as above) will always be \emph{linear} in the size of the universal channel set.
This is true even if  the available channel  set of all nodes contain a single common channel (an extreme case).
In many cases, the available channel sets of nodes may be much smaller than the universal channel set.
Third, all nodes should start executing the algorithm at the same time. Otherwise, different nodes may tune to different channels in the same time slot, thereby causing the multi-channel neighbor discovery algorithm to fail. As such, this transformation \emph{does not work for an asynchronous system}.
In comparison, in our algorithms, nodes do not need to agree on a universal channel set and the time complexity of our algorithms depends on the ``degree of heterogeneity'' in the network.

\item \emph{Single-Hop Network:}
Neighbor discovery algorithms in~\cite{AloKra+:2003:ADHOCNOW, AraVen+:2008:DYSPAN, AstZor+:2010:CCW} consider a single-hop network in which all nodes can communicate with each other directly. Our algorithms, on the other hand, work even for a multi-hop network.

\item \emph{Synchronous System:}
Neighbor discovery algorithms in \cite{IyePru+:2011:SASO,SunWu+:2012:GLOBECOM,ChaHsu:2013:ICCP} assume a time-slotted synchronous system.
Work in \cite{KriTho+:2008:CN,MitKri+:2009:JPDC,ZenMit+:2010:ISPDC} assume a synchronous system in which all nodes initiate neighbor discovery at the same time. Moreover, the proposed algorithms are deterministic in nature, and have high time complexity that depends on the 
\emph{product} of network size and universal channel set size. In contrast, our focus in this work is on developing algorithms for an asynchronous system that have lower running time (albeit provide only probabilistic guarantees).

\item \emph{Homogeneous Channel Availability:}
Neighbor discovery algorithms in \cite{AloKra+:2003:ADHOCNOW, KarVia+:2011:INFOCOM, ChaHsu:2013:ICCP} assume homogeneous channel availability (all channels are available to all nodes). Our algorithms, on the other hand, work even when different nodes have different subsets of channels available for communication, and, further, their running time adapts to the ``degree of heterogeneity'' in the network.

\item \emph{Little or No Theoretical Guarantees:}
Work in \cite{RajGan+:2004:GLOBECOM,WanSun+:2013:ICCT} only evaluate the proposed neighbor discovery algorithms using experiments; no attempt is made to
provide any theoretical (deterministic or probabilistic) guarantees. Work in \cite{AstZor+:2010:CCW} provides a suite of neighbor discovery algorithms that can tolerate jamming attacks. However, it does not provide any theoretical guarantees either (algorithms are evaluated using only simulation experiments). In our work, in contrast, we analyze all our algorithms theoretically and establish probabilistic bounds on their performance.

\end{itemize}

Besides the related work mentioned above,
Raniwala and Chiueh  propose a neighbor discovery algorithm for a multi-channel wireless network in~\cite{RanChi:2005:INFOCOM} but their work 
assumes that each node has multiple interfaces. 
Law \emph{et al.} describe a neighbor discovery algorithm for constructing a scatternet in~\cite{LawMeh+:2001:MOBIHOC}. Salondis \emph{et al.} describe an algorithm for \emph{two-node} link formation in~\cite{SalBha+:2000:MOBIHOC}, which is then used as a building block for forming a scatternet in~\cite{SalBha+:2001:INFOCOM}.
A scatternet is an ad hoc network consisting of two or more piconets; a piconet is an ad hoc network consisting two or more (up to eight) Bluetooth-enabled devices with one device acting as master and remaining devices acting as slaves.
Work in \cite{GanWan+:2012:SECON,LiuLin+:2012:TPDS,GuHua+:2013:SECON,BiaPar:2013:TMC,MisMis+:2013:VTC,ChaHsu:2013:ICCP,CheBia+:2014:MOBIHOC:a,CheBia+:2014:MOBIHOC:b} addresses the closely-related problem of \emph{rendezvous} between two CR nodes. Intuitively, the rendezvous problem between two nodes involves ensuring that the two nodes tune to the same channel at the same time, which is necessary for nodes to discover each other. 
However, as opposed to the problem considered in this work, the rendezvous problem does not consider the effect of possible simultaneous transmissions by other neighboring nodes at the same, which may cause collisions thereby preventing a link from being discovered.

%% file: conclusion.tex
\section{Conclusion}

In this work, we have proposed several randomized neighbor discovery 
algorithms for an \model{} network both for synchronous and asynchronous 
systems under a variety of assumptions. Our algorithms guarantee success with 
arbitrarily high probability. All our algorithms are designed to minimize the 
amount of knowledge that nodes need to possess \emph{a priori} before they 
start executing the neighbor discovery algorithm. We have shown that our neighbor discovery algorithms are robust to unreliable channels are adversarial attacks. 

As a future, we plan to develop randomized algorithms for solving other 
communication problems for an \model{} network such as broadcasting and 
convergecasting that guarantee success with arbitrarily high probability.

%% file: appendix.tex
\section{Omitted Proofs}

\subsection{Computing the probability of occurrence of $\A{\tau}{c}$} 
We have:
\begin{align*}
\nonumber
\Pr\{\A{\tau}{c}\} & =  (v \text{ selects } c \text{ at the beginning of } \tau) \land (v \text{ chooses to transmit during } \tau) \\
\nonumber
& =   \frac{1}{\card{\acs{v}}} \times \min\left( \frac{1}{2}, \frac{\card{\acs{v}}}{2^k}  \right) \\
\nonumber
& =  \min\left( \frac{1}{2\card{\acs{v}}}, \frac{1}{2^k} \right) \\
\nonumber
& \phantom{=} \big\{ \text{using \eqref{eq:slot}, } 2^{k-1} \leq \degreeon{u}{c}  \text{ which implies }  2^k \leq 2 \degreeon{u}{c } \big\} \\
\nonumber
& \geq  \frac{1}{\max\{ \: 2 \card{\acs{v}}, 2\degreeon{u}{c} \:\} } \\
\nonumber
&  \geq \frac{1}{2 \max(\acsmaximum,\dmaximum)}
\end{align*}

\subsection{Computing the probability of occurrence of  $\B{\tau}{c}$} 
We have:
\begin{align*}
\nonumber
\Pr\{\B{\tau}{c}\} & =  (u \text{ selects } c \text{ at the beginning of } \tau) \land (u \text{ chooses to listen during } \tau) \\
\nonumber
& =  \frac{1}{\card{\acs{u}}} \times \left\{1 - \min\left( \frac{1}{2}, \frac{\card{\acs{u}}}{2^k}  \right) \right\} \\
\nonumber
& \phantom{=} \big\{  \min( x, y ) \leq x \big\}  \\
\nonumber
& \geq  \frac{1}{\card{\acs{u}}} \times \left(1 - \frac{1}{2} \right) = \frac{1}{2\card{\acs{u}}}
\end{align*}

\subsection{Computing the probability of  occurrence of $\C{\tau}{c}$} 
Let $\incoming{u}{c}$ denote the set of neighbors of $u$ on $c$. Note that, if $\incoming{u}{c}$ only contains $v$, then $\Pr\{ \C{\tau}{c} \} = 1$. Otherwise, we have:
\begin{align*}
\nonumber
\Pr\{\C{\tau}{c}\}  
& =   \prod_{\substack{w \in \incoming{u}{c} \\ w \neq v}}  \Pr(w \text{ does not transmit on } c \text{ during } \tau) \\
\nonumber
& =  \prod_{\substack{w \in \incoming{u}{c} \\ w \neq v}}  \{1 - \Pr(w \text{ transmits on } c \text{ during } \tau)\} \\
\nonumber
& =  \prod_{\substack{w \in \incoming{u}{c} \\ w \neq v}}  \left\{ 1 - \frac{1}{\card{\acs{w}}} \times \min\left( \frac{1}{2}, \frac{\card{\acs{w}}}{2^k} \right) \right\} \\
\nonumber
& =   \prod_{\substack{w \in \incoming{u}{c} \\ w \neq v}}  \left\{ 1 - \min\left( \frac{1}{2\card{\acs{w}}}, \frac{1}{2^k} \right) \right\} \\
\nonumber
& \phantom{=} \big\{ \min(x, y) \leq y \big\} \\
\nonumber
& \geq  \prod_{\substack{w \in \incoming{u}{c} \\ w \neq v}} \!\!\!\! \left( 1 - \frac{1}{2^k} \right) \\
\nonumber
& =  \left(1 - \frac{1}{2^k} \right)^{\card{\incoming{u}{c}} - 1} \\
\nonumber
& \phantom{=} \big\{ \card{\incoming{u}{c}} - 1 =  \degreeon{u}{c} - 1 \text{ and, using \eqref{eq:slot}, } \degreeon{u}{c} - 1 \leq 2^k  \big\} \\
\nonumber 
& \geq  \left(1 - \frac{1}{2^k} \right)^{2^k} \\
\nonumber
& \phantom{=} \left\{ \forall x \geq 2,  \left( 1 - \tfrac{1}{x} \right)^x \text{ is a monotonically increasing function of } x  \text{ and thus } \geq \tfrac{1}{4}  \right\}  \\
\nonumber
& \geq  \frac{1}{4}
\end{align*}

\subsection{Computing the probability that  $s$ covers $\link{v}{u}$} 
Let $\F{\tau}{c}$ denote the event that  $\tau$ covers $\link{v}{u}$ on $c$. Note that $\F{\tau}{c} = \A{\tau}{c} \land \B{\tau}{c} \land \C{\tau}{c}$. Since $\A{\tau}{c}$, $\B{\tau}{c}$ and $\C{\tau}{c}$ are mutually independent events, $\Pr\{\F{\tau}{c}\} = \Pr\{\A{\tau}{c}\} \times \Pr\{\B{\tau}{c}\} \times \Pr\{\C{\tau}{c}\}$. We have:
\begin{align*}
\nonumber
\Pr\{\F{\tau}{c}\} 
& \geq \frac{1}{2 \max(\acsmaximum,\dmaximum)} \times \frac{1}{2\card{\acs{u}}} \times \frac{1}{4} \\
\nonumber
& = \frac{1}{16 \card{\acs{u}} \max(\acsmaximum,\dmaximum)}
\end{align*}

Now, we have:
\begin{align*}
\nonumber
& \Pr\{ s \text{ covers } \link{v}{u} \} \\
\nonumber
& = \sum_{c} \Pr\{ s \text{ covers } \link{v}{u} \text{ on } c  \} \\
\nonumber
& = \sum_{c \in \acs{v} \cap \acs{u}}  \Pr\{ s \text{ covers } \link{v}{u} \text{ on } c  \}  + \sum_{c \notin \acs{v} \cap \acs{u}}  \Pr\{ s \text{ covers } \link{v}{u} \text{ on } c  \}  \\
\nonumber
& \phantom{=} \big\{ c \notin \acs{v} \cap \acs{u}  \text{ implies that } \Pr\{ s \text{ covers } \link{v}{u} \text{ on } c  \} = 0 \big\} \\
\nonumber
& = \sum_{c \in \acs{v} \cap \acs{u}}  \Pr\{ s \text{ covers } \link{v}{u} \text{ on } c  \}   \\
\nonumber
& \geq \sum_{c \in \acs{v} \cap \acs{u}}  \Pr\{\F{\tau}{c}\}  \\
\nonumber
& \geq \sum_{c \in \acs{v} \cap \acs{u}}  \frac{1}{16 \card{\acs{u}} \max(\acsmaximum,\dmaximum)}  \\
\nonumber
& \quad = \quad \frac{\card{\acs{v} \cap \acs{u}}}{\card{\acs{u}}} \times \frac{1}{16\max(\acsmaximum,\dmaximum)} 
\geq \frac{\srminimum}{16 \max(\acsmaximum,\dmaximum)}
\end{align*}

\subsection{Computing the probability of successful discovery of a link in $\frac{16 \max( \acsmaximum, \dmaximum )}{\srminimum} \ln \left( \frac{N^2}{\epmaximum} \right)$ stages}
Let $M = \frac{16 \max( \acsmaximum, \dmaximum )}{\srminimum} \ln \left( \frac{N^2}{\epmaximum} \right)$. We have:
\begin{align*}
\nonumber
& \Pr(\link{v}{u} \text{ is not covered within } M \text{ stages}) \\
\nonumber 
& = \prod_{1 \leq i \leq M} \Pr(\link{v}{u} \small \text{ is not covered in the } i^{th} \text{ stage})  \\
\nonumber 
& \leq  \prod_{1 \leq i \leq M}  \left( 1 - \frac{\srminimum}{16\max(\acsmaximum,\dmaximum)} \right) \\
\nonumber
& =  \left( 1 - \frac{\srminimum}{16\max(\acsmaximum,\dmaximum)} \right)^{ \frac{16 \max( \acsmaximum, \dmaximum )}{\srminimum} \ln \left( \frac{N^2}{\epmaximum} \right) } \\
\nonumber
& \phantom{=} \left\{ \forall x \geq 2, \left(1 - \tfrac{1}{x}\right)^x \text{ is upper bounded by }  \tfrac{1}{e} \right\} \\
\nonumber
&  \leq \left(\frac{1}{e}\right)^{\ln \left( \frac{N^2}{\epmaximum} \right)} =  \frac{\epmaximum}{N^2} 
\end{align*}

\subsection{Computing the sum of the number of time-slots/frames in the first $k$ epochs}
Let $S$ denote the sum $\sum_{1 \leq i \leq k} (i+1) \cdot 2^i$. We have:
\[
\begin{array}{r@{\quad=\quad}rcrcrcrcrcr}
S & 2 \cdot 2^1 & + & 3 \cdot 2^2 & + & 4 \cdot 2^3 & + & \cdots & + & (k+1) \cdot 2^k  \\
2S &            &   & 2 \cdot 2^2 & + & 3 \cdot 2^3 & + & \cdots & + & k \cdot 2^k  & + & (k+1) \cdot 2^{k+1} 
\end{array}
\]
Subtracting the first equation from the second and rearranging terms, we obtain:
\begin{align*}
\nonumber
S 
& = (k+1) \cdot 2^{k+1} - (2^k + 2^{k-1} + \cdots + 2^2) - 2 \cdot 2^1 \\
\nonumber
& = (k+1) \cdot 2^{k+1} - \left( \sum_{1 \leq i \leq k} 2^i \right)  - 2 \\
\nonumber
& = (k+1) \cdot 2^{k+1} - 2 \cdot (2^k - 1) - 2 \\
& = k \cdot 2^{k+1} 
\label{eq:epoch|sum}
\end{align*}

\subsection{Computing the probability of occurrence of $\CFS{f}{g}{c}$}
\label{sec:CC|slot}

We focus on the case when exactly three slots of $w$ overlap with $\slot{f}{g}$. But the result also holds if fewer than three slots of $w$ overlap with $\slot{f}{g}$.  There are two sub-cases: the three (overlapping) slots belong to the same frame or two different frames.

\paragraph{Sub-case 1 (three slots belong to the same frame):} In this case, we have:
\begin{align*}
\Pr\{ w \text{ does not transmit on } c \text{ during } \slot{f}{g} \} \}
& = 1 - p + p \cdot \left( \frac{\acsmaximum-1}{\acsmaximum} \right)^3 \\
& = 1 - p \left( 1 - \frac{(\acsmaximum-1)^3}{\acsmaximum^3} \right) \\
& = 1 - p \left( \frac{3\acsmaximum^2 - 3\acsmaximum + 1}{\acsmaximum^3} \right) \\
& \geq 1 - p \left( \frac{3\acsmaximum^2}{\acsmaximum^3} \right) \\
& \geq 1 - \left( \frac{\acsmaximum}{3 \destimate} \right) \cdot \left( \frac{3}{\acsmaximum} \right) \\
& = 1 - \frac{1}{\destimate}
\end{align*}

\paragraph{Sub-case 2 (the three slots come from two different frames):} In this case, we have:
\begin{align*}
\Pr\{ w \text{ does not transmit on } c \text{ during } \slot{f}{g} \}
& = \left\{ 1 - p + p \cdot \left(\frac{\acsmaximum-1}{\acsmaximum}\right) \right\} \cdot \left\{1 - p + p \cdot \left(\frac{\acsmaximum-1}{\acsmaximum}\right)^2 \right\} \\
& = \left( 1 - p \cdot \frac{1}{\acsmaximum} \right) \cdot \left( 1 - p \cdot \frac{2\acsmaximum-1}{\acsmaximum^2} \right) \\
& = 1 - p \cdot \frac{1}{\acsmaximum} - p \cdot \frac{2\acsmaximum-1}{\acsmaximum^2} + p^2 \cdot \frac{2\acsmaximum - 1}{\acsmaximum^3} \\
& \geq 1 - p \cdot \frac{3\acsmaximum-1}{\acsmaximum^2} \\
& \geq 1 - p \cdot \frac{3 \acsmaximum}{\acsmaximum^2} \\
& = 1 - p \cdot \frac{3}{\acsmaximum} \\
& \geq 1 - \frac{\acsmaximum}{3 \destimate} \cdot \frac{3}{\acsmaximum} = 1 - \frac{1}{\destimate} 
\end{align*}

It can now be shown, in a similar manner as in the proof of \lemref{legal|atleast}, that $\Pr\{ \CFS{f}{g}{c} \} \geq \frac{1}{4}$.